%% file: main.tex
\documentclass[12pt]{llncs}

\usepackage{fullpage}
\usepackage{xspace}
\usepackage[utf8]{inputenc}
\usepackage{enumitem}
\usepackage{cite}

\input{preamble}

\title{Quantitative Reductions and \newline{} Vertex-Ranked Infinite Games%
\thanks{Supported by by the project ``TriCS'' (ZI 1516/1-1) of the German Research Foundation (DFG) and the Saarbrücken Graduate School of Computer Science. This work was mainly performed while the author was employed at the Reactive Systems Group, Saarland University, Germany.}
\thanks{This work is based on work first presented at GandALF '18~\cite{Weinert18}. %
	In particular, the proofs of Theorem~\ref{thm:quantitative-reductions:transitivity}, Theorem~\ref{thm:vertex-ranked-games:direct:complexity}, and Lemma~\ref{lem:request-response:reduction} have been revised and extended in comparison to that version. %
	Furthermore, Theorem~\ref{thm:vertex-ranked-games:generalization:sup}, Remark~\ref{rem:vertex-ranked-games:generalization:lim}, and Section~\ref{sec:applications:muller} have been added.
	Finally, the proofs of all lemmas and theorems are now included in the main text.}
}

\author{Alexander Weinert}
\institute{German Aerospace Center (DLR), Institute for Software Technology, Linder H{\"o}he, 51147 K{\"o}ln, Germany \\ \email{alexander.weinert@dlr.de}}

\begin{document}
\maketitle

\input{abstract}

\input{introduction}

\input{preliminaries}

\input{quantitative-reductions}

\input{vertex-ranked-games}

\input{applications}

\input{conclusion}

\bibliographystyle{splncs03}
\bibliography{main}

\end{document}

%% file: preamble.tex
\usepackage{amsmath}
\usepackage{amsfonts}
\usepackage{amssymb}
\usepackage{xstring}
\usepackage{url}
\usepackage{tabularx}
\usepackage{xspace}
\usepackage{hyperref}

\usepackage{multirow}

\usepackage{tikz}
\usetikzlibrary{decorations.pathreplacing}
\usetikzlibrary{decorations.pathmorphing}
\usetikzlibrary{shapes,shapes.symbols,automata,arrows}
\usetikzlibrary{calc}
\usetikzlibrary{patterns}

\tikzset{p0/.style = {shape = circle, draw, thick, minimum size = 0.7cm}}
\tikzset{p1/.style = {rectangle, minimum size=.7cm, draw, thick}}
\tikzset{>=stealth, shorten >=1pt}
\tikzset{every edge/.style = {thick, ->, draw}}
\tikzset{every loop/.style = {thick, ->, draw}}

\makeatletter
\tikzset{circle split part fill/.style  args={#1,#2}{%
 alias=tmp@name, 
  postaction={%
    insert path={
     \pgfextra{%
     \pgfpointdiff{\pgfpointanchor{\pgf@node@name}{center}}%
                  {\pgfpointanchor{\pgf@node@name}{east}}%
     \pgfmathsetmacro\insiderad{\pgf@x}
      \fill[#1] (\pgf@node@name.base) ([xshift=-\pgflinewidth]\pgf@node@name.east) arc
                          (0:180:\insiderad-\pgflinewidth)--cycle;
      \fill[#2] (\pgf@node@name.base) ([xshift=\pgflinewidth]\pgf@node@name.west)  arc
                           (180:360:\insiderad-\pgflinewidth)--cycle; 
         }}}}}  
 \makeatother

\usepackage{booktabs}

\usepackage{pifont}

\definecolor{myred}{rgb}{1,0.604,0.604}
\definecolor{mydarkred}{rgb}{1,0.345,0.345}
\definecolor{myblue}{rgb}{0.635,0.675,0.966}
\definecolor{mydarkblue}{rgb}{0.412,0.475,0.957}
\definecolor{myyellow}{rgb}{1,0.976,0.604}
\definecolor{mydarkyellow}{rgb}{1,0.961,0.345}

\tikzset{
	assign/.style = { fill=myblue },
	choice/.style = { fill=myred },
	check/.style  = { fill=myyellow }
}

\makeatletter
\DeclareRobustCommand{\rvdots}{%
  \vbox{
    \baselineskip4\p@\lineskiplimit\z@
    \kern-\p@
    \hbox{.}\hbox{.}\hbox{.}
  }}
\makeatother

\newcommand\restr[2]{{
  \left.\kern-\nulldelimiterspace 
  #1 
  \vphantom{\big|} 
  \right|_{#2} 
  }}

\newcommand{\ext}{\mathrm{ext}}

\newcommand{\rank}{\mathrm{rk}}

\newcommand{\subarenaof}{\sqsubseteq}

\newcommand{\declim}{\text{dec}_{\lim}}
\newcommand{\decsup}{\text{dec}_{\sup}}

\newcommand{\initmark}{I}

\newcommand{\vinit}{v_{\initmark}}

\newcommand{\myquot}[1]{``#1''}

\newcommand{\bigo}[0]{\mathcal{O}}

\newcommand{\size}[1]{|#1|}
\newcommand{\card}[1]{\size{#1}}
\newcommand{\set}[1]{\{ #1 \}}
\newcommand{\nats}{\mathbb{N}}

\renewcommand{\epsilon}{\varepsilon}

\newcommand{\arena}{\mathcal{A}}
\newcommand{\game}{\mathcal{G}}
\newcommand{\gamelim}{\mathcal{G}_{\lim}}
\newcommand{\gamesup}{\mathcal{G}_{\sup}}
\newcommand{\cost}{\mathrm{Cost}}

\newcommand{\wincond}{\mathrm{Win}}
\newcommand{\win}{\mathrm{W}}
\newcommand{\plays}{\mathrm{Plays}}

\newcommand{\equivclass}[2]{[#1]_{#2}}

\newcommand{\quotient}[2]{{#1}\big/#2}
\newcommand{\val}{\mathrm{val}}

\newcommand{\reducesto}{\leq}

\newcommand{\safety}{\mathrm{Safety}}

\newcommand{\mem}{\mathcal{M}}
\newcommand{\init}{m_\initmark}
\newcommand{\update}{\mathrm{Upd}}
\newcommand{\nxt}{\mathrm{Nxt}}

\newcommand{\quantmuller}{\textsc{QuantMuller}}
\newcommand{\muller}{\textsc{Muller}}
\newcommand{\mullerfam}{{\mathcal F}}
\newcommand{\score}{\text{Score}}
\newcommand{\acc}{\text{Acc}}

\newcommand{\reqres}{\textsc{ReqRes}}
\newcommand{\costreqres}{\textsc{CostReqRes}}
\newcommand{\reqresdist}{\textsc{ReqResCor}}

\newcommand{\att}[3]{\mathrm{Attr}_{#1}^{#2}(#3)}

\newcommand{\exptime}{\textsc{ExpTime}}

\newcommand{\pspace}{\textsc{PSpace}\xspace}

%% file: abstract.tex
\begin{abstract}
	We introduce quantitative reductions, a novel technique for structuring the space of quantitative games and solving them that does not rely on a reduction to qualitative games.
	We show that such reductions exhibit the same desirable properties as their qualitative counterparts and that they additionally retain the optimality of solutions.
	Moreover, we introduce vertex-ranked games as a general-purpose target for quantitative reductions and show how to solve them.
	In such games, the value of a play is determined only by a qualitative winning condition and a ranking of the vertices.
	
	We provide quantitative reductions of quantitative request-response games and of quantitative Muller games to vertex-ranked games, thus showing \exptime-completeness of solving the former two kinds of games.
	In addition, we exhibit the usefulness and flexibility of vertex-ranked games by showing how to use such games to compute fault-resilient strategies for safety specifications.
	This work lays the foundation for a general study of fault-resilient strategies for more complex winning conditions.
\end{abstract}

%% file: introduction.tex
\section{Introduction}
\label{sec:introduction}

The study of quantitative infinite games has garnered great interest lately, as they allow for a much more fine-grained analysis and specification of reactive systems than classical qualitative games~\cite{BrimCDGR11,BruyereFiliotRandourRaskin17,BruyereHautemRandour16,EhrenfeuchtMycielski79,FZ14,WeinertZimmermann17,ZwickPaterson95}.
While there exists previous work investigating quantitative games, the approaches to solving them usually rely on ad-hoc solutions that are tailor-made to the considered winning condition.
Moreover, quantitative games are usually solved by reducing them to a qualitative game in a first step, hardcoding a certain value of interest during the reduction.
In particular, to the best of our knowledge, there exists no general framework for the analysis of such games that is analogous to the existing one for qualitative games.
In this work, we introduce such a framework that disentangles the study of quantitative games from that of qualitative ones.

Qualitative infinite games have been applied successfully in the verification and synthesis of reactive systems~\cite{BloemChatterjeeEtAl14,BloemChatterjeeHenzingerJobstmann09,BouyerMarkeyOlschewskiUmmels11,CernyChatterjeeHenzingerRadhakrishnaSingh11}.
They have given rise to a multitude of algorithms that ascertain system correctness and that synthesize correct-by-construction systems.
In such a game, two players, called Player~$0$ and Player~$1$, move a token in a directed graph.
After infinitely many moves, the resulting sequence of vertices is evaluated and one player is declared the winner of the play.
For example, in a qualitative request-response game~\cite{WallmeierHuettenThomas03}, the goal for Player~$0$ is to ensure that every visit to a vertex denoting some request is eventually followed by a visit to a vertex denoting an answer to that request.
In order to solve qualitative games, i.e., to determine a winning strategy for one player, one often reduces a complex game to a potentially larger, but conceptually simpler one.
For example, in a multi-dimensional request-response game, i.e., in a request-response game in which there exist multiple conditions that can be requested and answered, one stores the set of open requests and demands that every request is closed at infinitely many positions.
As this is a Büchi condition, which is much simpler than the request-response condition, one is able to reduce request-response games to Büchi games.

In recent years the focus of research has shifted from the study of qualitative games, in which one player is declared the winner of a given play, to that of quantitative games, in which the resulting play is assigned some value or cost.
Such games allow, for example, modeling systems in which requests have to be answered within a certain number of steps~\cite{BruyereHautemRandour16,ChatterjeeHenzingerHorn09,FaymonvilleZimmermann14,FZ14,KupfermanPitermanVardi09}, systems with one or more finite resources which may be drained and charged~\cite{BouyerMarkeyRandourLarsenLaursen16,ChatterjeeDoyen12,ChatterjeeDoyenHenzingerRaskin10,DBLP:journals/corr/abs-1209-3234}, or scenarios in which each move incurs a certain cost for either player~\cite{EhrenfeuchtMycielski79,ZwickPaterson95}.

In general, Player~$0$ aims to minimize the cost of the resulting play, i.e., to maximize its value, while Player~$1$ seeks to maximize the cost, thus minimizing the value.
In a quantitative request-response game, for example, it is the goal of Player~$0$ to minimize the number of steps taken between requests and their corresponding answers.
The typical questions asked in the context of such games are
\myquot{Does there exist an upper bound on the time between requests and responses that Player~$0$ can ensure?}~\cite{ChatterjeeHenzingerHorn09,FZ14,KupfermanPitermanVardi09,ScheweWeinertZimmermann18},
\myquot{Can Player~$0$ ensure an average cost per step greater than zero?}~\cite{ZwickPaterson95},
\myquot{What is the minimal time between requests and responses that Player~$0$ can ensure?}~\cite{WeinertZimmermann17}, or
\myquot{What is the minimal average level of the resource that Player~$0$ can ensure without it ever running out?}~\cite{BouyerMarkeyRandourLarsenLaursen16}.
The former two questions can be seen as boundedness questions, while the latter two are asking for optimal solutions.

Such decision problems are usually solved by fixing some bound~$b$ on the cost of the resulting plays and subsequently reducing the problem of finding a strategy for Player~$0$ that enforces a cost of at most~$b$ in the quantitative game to the problem of solving a qualitative game, hardcoding the fixed~$b$ in the process.
The problem of deciding whether or not Player~$0$ has such a strategy is called the ($b$-)threshold problem.

For example, in order to determine the winner in a quantitative request-response game as described above for some bound~$b$, we construct a Büchi game in which every time a request is opened, a counter for that request is started which counts up to the bound~$b$ and is reset if the request is answered.
Once any counter exceeds the value~$b$, we move to a terminal position indicating that Player~$0$ has lost.
We then require that every counter is inactive infinitely often, which is again a Büchi condition and thus much simpler than the original quantitative request-response condition.
Thus, Player~$0$ wins the resulting qualitative game if and only if she can ensure that every request is answered within at most~$b$ steps in the quantitative game.

Such reductions are usually specific to the problem being addressed.
Furthermore, they immediately abandon the quantitative aspect of the game under consideration, as the bound is hardcoded during the first step of the analysis.
Thus, even when only changing the bound one is interested in, the reduction has to be recomputed and the resulting qualitative game has to be solved from scratch.
In our request-response example, if one is interested in deciding the~$b'$-threshold problem for some~$b' \neq b$, one constructs a new Büchi game for the bound~$b'$.
This game is then solved independently of the one previously computed for the bound~$b$.

In this work, we lift the concept of reductions for qualitative games to quantitative games.
Such quantitative reductions enable the study of a multitude of optimization problems for quantitative games in a way similar to decision problems for qualitative games.
When investigating quantitative request-response games using quantitative reductions, for example, we only compute a single, simpler quantitative game and subsequently check this game for a winning strategy for Player~$0$ for any bound~$b$.
If she has such a strategy in the latter game, the quantitative reduction yields a strategy for her satisfying the same bound in the former one.

In general, we retain the intuitive property of reductions for qualitative games:
Using quantitative reductions, the properties of a complex quantitative game can be studied by investigating a potentially larger, but conceptually simpler quantitative game.

\paragraph*{Contributions}
\label{sec:introduction:contributions}
We present the first framework for reductions between quantitative games and we provide vertex-ranked games as general-purpose targets for such reductions.
Moreover, we show tight bounds on the complexity of solving vertex-ranked games with respect to a given bound.

Subsequently, we provide three examples illustrating the use of the concepts introduced in this work:
First, we define quantitative request-response games and solve them using quantitative reductions to vertex-ranked games.
Second, we show how to solve quantitative Muller games as defined by McNaughton~\cite{McNaughton00} via quantitative reductions to vertex-ranked safety games.
Third, we illustrate the versatility of vertex-ranked games by using them to compute fault-resilient strategies for safety games with faults.
We summarize our contributions with regards to solving quantitative games in Table~\ref{tab:contributions:solving} and we summarize our contributions with regards to solving quantitative games optimally in Table~\ref{tab:contributions:optimization}.

\begin{table}
\centering \footnotesize
\begin{tabular}{lccc} \toprule
\multicolumn{1}{c}{Game}%
& Time & Space & Memory \\ \midrule
\parbox{.3\textwidth}{Vertex-ranked \\ \hspace*{1em} $\sup$-games}%
	& $\bigo(n) + t(\card{\game})$%
	& $\bigo(n) + s(\card{\game})$%
	& $\bigo(\card{\sigma})$ \\ \midrule
\parbox{.3\textwidth}{Vertex-ranked \\ \hspace*{1em} $\limsup$-games}%
	& $\bigo(n^3 + n^2\cdot t(\card{\game}))$%
	& $\bigo(n + s(\card{\game}))$
	& $\bigo(\card{\sigma})$ \\ \midrule
\parbox{.3\textwidth}{Quantitative \\ \hspace*{1em} request-response games}%
	& $\bigo( n^2 b^{2d} d^2 2^d )$%
	& -- %
	& $\bigo(n b^d d 2^d)$ \\ \midrule
\parbox{.3\textwidth}{Quantitative \\ \hspace*{1em} Muller games}%
	& $\bigo((n!)^3)$%
	& -- %
	& $\bigo((n!)^3)$ \\ \bottomrule
\end{tabular}
\caption{%
	A summary of the proofs in work concerned with solving games. %
	For each game,~$n$ denotes the number of vertices of the game. %
	For vertex-ranked games,~$\game$ denotes the underlying qualitative game while~$t(\card{\game})$ and~$s(\card{\game})$ denote the time and space required to solve that underlying game. %
	Furthermore,~$\sigma$ denotes the size of a winning strategy in the corresponding qualitative game. %
	For request-response games,~$d$ and~$W$ denote the number of request-response pairs and the largest weight assigned to any edge in the game, respectively. %
	For the sake of brevity and consistency, we use the shorthand $b = d2^dnW$.
}
\label{tab:contributions:solving}
\end{table}

\begin{table}
\centering \footnotesize
\begin{tabular}{lccc} \toprule
\multicolumn{1}{c}{Game}%
& Time & Space \\ \midrule
\parbox{.3\textwidth}{Vertex-ranked \\ \hspace*{1em} $\sup$-games}%
	& $\bigo(\log(M)(n + t(\card{\game})))$%
	& $\bigo(M) + s(\card{\game})$ \\ \midrule
\parbox{.3\textwidth}{Vertex-ranked \\ \hspace*{1em} $\limsup$-games}%
	& $\bigo(\log(M)(n^3 + n^2\cdot t(\card{\game})))$%
	& $\bigo(n + s(\card{\game}))$ \\ \midrule
\parbox{.3\textwidth}{Quantitative \\ \hspace*{1em} request-response games}%
	& $\bigo(\log(b) ( n^2 b^{2d} d^2 2^d ) )$ %
	& -- \\ \midrule
\parbox{.3\textwidth}{Quantitative \\ \hspace*{1em} Muller games}%
	& $\bigo((n!)^3)$ %
	& -- \\ \bottomrule
\end{tabular}
\caption{%
	A summary of the proofs in work concerned with solving games optimally. %
	For vertex-ranked games,~$M$ denotes the number of different ranks assigned to vertices in the game. %
	The remaining symbols are defined as in Table~\ref{tab:contributions:solving}.
}
\label{tab:contributions:optimization}
\end{table}

\paragraph*{Structure of this work}
\label{sec:introduction:structure}
After introducing qualitative and quantitative games formally in Section~\ref{sec:preliminaries}, we define quantitative reductions in Section~\ref{sec:quantitative-reductions} and show that they provide a mechanism to determine the minimal bound~$b$ such that Player~$0$ can enforce a cost of at most~$b$ in a given quantitative game:
If a game~$\game$ can be reduced to a game~$\game'$, then we can use a strategy for Player~$0$ that minimizes the cost of plays in~$\game'$ to construct a strategy for her which minimizes the cost of plays in~$\game$.

In Section~\ref{sec:vertex-ranked-games}, we define vertex-ranked games, very general classes of quantitative games that can be used as targets for quantitative reductions.
The quantitative condition of such games is quite simple in that the cost of a play is determined only by a qualitative winning condition and by a ranking of the vertices of the game.
If the resulting play is winning according to the qualitative condition, then its cost is given by the highest rank visited at all or visited infinitely often, depending on the particular variant of vertex-ranked games.
Otherwise, the value of the play is infinite.
We show that solving such vertex-ranked games is as hard as solving games with the underlying qualitative winning condition.

Finally, in Section~\ref{sec:application} we provide three examples of the versatility of vertex-ranked games:
First, we define and solve request-response games with costs via quantitative reductions.
Second, we recall the definition of quantitative Muller games due to McNaughton~\cite{McNaughton00} and show how to solve these games via quantitative reductions as well.
Third, we discuss how to use vertex-ranked games to compute fault-resilient strategies in safety games with faults~\cite{DallalNeiderTabuada16}.
In such games, after Player~$0$ has picked a move, say to vertex~$v$, a fault may occur, which overrides the choice of Player~$0$ and the game continues in vertex~$v' \neq v$ instead.
By using vertex-ranked games, we are able to compute strategies that are resilient against as many faults as possible.

%% file: preliminaries.tex
\section{Preliminaries}
\label{sec:preliminaries}

We first define notions that are common to both qualitative and quantitative games.
Afterwards, we recapitulate the standard notions for qualitative games before defining quantitative games and lifting the notions for qualitative games to the quantitative case.

We denote the non-negative integers by $\nats$ and define $[n] = \set{0, 1, \ldots, n-1}$ for every $n \ge 1$.
Also, we define~$\infty > n$ for all~$n \in \nats$ and~$\nats_\infty = \nats \cup \set{\infty}$.
Finally, for any set~$V$, we write~$\card{V}$ and~$V^\omega$ to denote the cardinality of~$V$ and the set of all infinite sequences over~$V$, respectively.
An \textit{arena}~$\arena=(V, V_0, V_1, E, \vinit)$ consists of a finite, directed graph~$(V, E)$, a partition~$(V_0, V_1)$ of $V$ into the vertices of Player~$0$ and Player~$1$, and an initial vertex~$\vinit \in V$. The size of $\arena$, denoted by $\size{\arena}$, is defined as $\size{V}$.
A \textit{play} in $\arena$ is an infinite path~$\rho = v_0 v_1 v_2 \cdots$ through $(V, E)$ starting in $\vinit$.
To rule out finite plays, we require every vertex to be non-terminal, i.e., we require that for every vertex~$v \in V$ there exists some vertex~$v' \in V$ such that~$(v, v') \in E$.

A \textit{strategy} for Player~$i$ is a mapping $\sigma \colon V^*V_i \rightarrow V$ that assigns to each play prefix~$\pi$ ending in a vertex of Player~$i$ a vertex~$\sigma(\pi)$ to move to.
Formally, we require $(v, \sigma(\pi v)) \in E$ for all $\pi \in V^*$, $v \in V_i$.
We say that $\sigma$ is \textit{positional} if $\sigma(\pi v) = \sigma(v)$ for every $\pi \in V^*$, $v \in V_i$.
A play $v_0 v_1 v_2 \cdots$ is \textit{consistent} with a strategy~$\sigma$ for Player~$i$, if $v_{j+1} = \sigma( v_0 \cdots v_j)$ for all~$j$ with $v_j \in V_i$.

A \textit{memory structure}~$\mem = (M, \init, \update)$ for an arena $(V, V_0, V_1, E, \vinit)$ consists of a finite set~$M$ of memory states, an initial memory state $\init \in M$, and an update function~$\update\colon M \times V \rightarrow M$.
We extend the update function to finite play prefixes in the usual way:
$\update^+(v_\initmark) = m_\initmark$ and $\update^+(\pi v) = \update(\update^+(\pi), v)$ for play prefixes $\pi \in V^+$ and $v \in V$.
A next-move function $\nxt \colon V_i \times M \rightarrow V$ for Player~$i$ has to satisfy $(v, \nxt(v, m)) \in E$ for all $v \in V_i$, $m \in M$.
Each pair of a memory structure~$\mem$ and a next-move function~$\nxt$ induces a strategy~$\sigma$ for Player~$i$ with memory~$\mem$ via $\sigma(v_0\cdots v_j) = \nxt(v_j, \update^+(v_0 \cdots v_j))$.
A strategy is called \textit{finite-state} if it can be implemented by a memory structure.
We define $\card{\mem} = \card{M}$.
In a slight abuse of notation, the size~$\card{\sigma}$ of a finite-state strategy is the size of a memory structure implementing it.

An arena $\arena = (V, V_0, V_1, E, \vinit)$ together with a memory structure $\mem = (M, \init, \update)$ for $\arena$ induce the extended arena 
\[
	\arena\times\mem = (V \times M, V_0 \times M, V_1 \times M, E', (\vinit, \init)) \enspace ,
\]
where~$E'$ is defined via $((v,m), (v',m')) \in E'$ if and only if $(v,v') \in E$ and $\update(m, v' ) = m'$.
Every play $\rho = v_0 v_1 v_2\cdots$ in $\arena$ has a unique extended play $\ext_\mem(\rho) = (v_0, m_0) (v_1, m_1)
(v_2, m_2) \cdots$ in $\arena \times \mem$ defined by $m_0 = \init$ and $m_{j+1} = \update(m_j, v_{j+1})$, i.e., $m_j = \update^+(v_0 \cdots v_j)$.
We omit the index~$\mem$ if it is clear from the context.
The extension of a finite play prefix in $\arena$ is defined analogously.

Let~$\arena$ be an arena, let $\mem_1 = (M_1, m^1_\initmark, \update_1)$ be a memory structure for~$\arena$, and let $\mem_2 = (M_2, m^2_\initmark, \update_2)$ be a memory structure for~$\arena \times \mem_1$.
We define $\mem_1 \times \mem_2 = (M_1 \times M_2, (m^1_\initmark, m^2_\initmark), \update)$, where $\update((m_1, m_2), v) = (m'_1,m'_2)$ if $\update_1(m_1, v) = m'_1$ and $\update_2(m_2, (v, m'_1)) = m'_2$.
Via a straightforward induction and in a slight abuse of notation we obtain
\[ \update^+_{\mem_1 \times \mem_2}(\pi) = (
	\update^+_{\mem_1}(\pi),
	\update^+_{\mem_2}(\ext_{\mem_1}(\pi)) )  \]
for all finite play prefixes~$\pi$, where we do not distinguish between the terms $(v, m_1, m_2)$, $((v, m_1), m_2)$, and $(v, (m_1, m_2))$.

\subsection{Qualitative Games}
\label{sec:preliminaries:qualitative}

A \textit{qualitative game}~$\game = ( \arena, \wincond )$ consists of an arena $\arena$ with vertex set~$V$ and a set~$\wincond \subseteq (V')^\omega$ of winning plays for Player~$0$, for some superset $V' \supseteq V$.
We call~$\wincond$ the winning condition of~$\game$.
In, e.g., a request-response game as described in Section~\ref{sec:introduction}, the winning condition contains the plays~$\rho$ in which for every visit to a vertex denoting a request there is a subsequent visit to a vertex denoting a response.
The set of winning plays for Player~$1$ is $V^\omega \setminus \wincond$.
As our definition of games is very general, the infinite object~$\wincond$ may not be finitely describable.
If it is, however, we slightly abuse notation and define~$\card{\game}$ as the sum of~$\card{\arena}$ and the size of a description of~$\wincond$.
In, e.g., the case of a request-response game, the set~$\wincond$ can be described via two sets~$Q, P \subseteq V$, where~$Q$ and~$P$ denote the vertices representing requests and responses, respectively.
We omit, however, a general definition of the term ``description of~$\wincond$'', as it is irrelevant for this work.

A strategy~$\sigma$ for Player~$i$ is a \textit{winning strategy} for her in $\game = (\arena, \wincond)$ if all plays consistent with~$\sigma$ are winning for her.
If Player~$i$ has a winning strategy, then we say she wins $\game$.
\textit{Solving} a game amounts to determining its winner, if one exists.
A game is \textit{determined} if one player has a winning strategy.

\subsection{Quantitative Games}
\label{sec:preliminaries:quantitative}

Quantitative games extend the classical model of qualitative games.
In a quantitative game, plays are not partitioned into winning and losing plays, but rather they are assigned some measure of quality.
We keep this definition very general in order to encompass many of the already existing models.
In Section~\ref{sec:vertex-ranked-games}, we give concrete examples of such games and show how to solve them optimally.

A \textit{quantitative game}~$\game = (\arena, \cost)$ consists of an arena~$\arena$ with vertex set~$V$ and a cost-function $\cost\colon (V')^\omega \rightarrow \nats_\infty$ for plays where $V' \supseteq V$ is again some superset of~$V$.
Similarly to the winning condition in the qualitative case,~$\cost$ is, in general, an infinite object.
If it is finitely describable, we, again slightly abusively, define the \textit{size}~$\card{\game}$ of~$\game$ as the sum of~$\card{\arena}$ and the size of a description of~$\cost$.
A cost function may, e.g., be described via a finite Mealy machine~\cite{Mealy55} that processes the play and outputs a natural number on each step.
The cost of a play could then be defined as the supremum, infimum, or average of the resulting sequence of numbers.
In such a case, a description of the cost function would consist of the Mealy machine, while the size of the description of the function could be defined as the sum of number of states of the machine and the length of the largest integer that is output by the machine.

A play~$\rho$ in~$\arena$ is \textit{winning} for Player~$0$ in~$\game$ if~$\cost(\rho) < \infty$.
Winning strategies, the winner of a game, and solving a game are defined as in the qualitative case.
In order to simplify the presentation, we only consider the case in which Player~$0$ aims to minimize the cost of a play. All concepts in this work can, however, be easily adapted to the dual case in which Player~$0$ aims to maximize the cost of a play.

We extend the cost-function over plays to strategies by defining~$\cost(\sigma) = \sup_\rho\cost(\rho)$ and~$\cost(\tau) = \inf_\rho\cost(\rho)$, where~$\rho$ ranges over the plays consistent with the strategy~$\sigma$ for Player~$0$ and over the plays consistent with the strategy~$\tau$ for Player~$1$, respectively.
Moreover, we say that a strategy~$\sigma$ for Player~$0$ is \textit{optimal} if its cost is minimal among all strategies for her.
Dually, a strategy~$\tau$ for Player~$1$ is optimal for him if its cost is maximal among all strategies for him.

For every strategy~$\sigma$ for Player~$0$, $\cost(\sigma) < \infty$ implies that~$\sigma$ is winning for Player~$0$.
However, the converse does not hold true:
Each play consistent with some strategy~$\sigma$ may have finite cost, while for every~$n \in \nats$ there exists a play~$\rho$ consistent with~$\sigma$ with~$\cost(\rho) \geq n$.
In contrast, a strategy~$\tau$ for Player~$1$ has~$\cost(\tau) = \infty$, if and only if~$\tau$ is winning for him.

We say that Player~$0$ wins~$\game$ with respect to~$b$ if she has a strategy~$\sigma$ with~$\cost(\sigma) \leq b$.
Dually, if Player~$1$ has a strategy~$\tau$ with~$\cost(\tau) > b$, then we say that he wins~$\game$ with respect to~$b$.
Solving a quantitative game~$\game$ with respect to~$b$ amounts to deciding whether or not Player~$0$ wins~$\game$ with respect to~$b$.
We call this decision problem the~\textit{$b$-threshold problem} of~$\game$ and omit the~$b$ if it is clear from the context.

If Player~$0$ has a strategy~$\sigma$ with $\cost(\sigma) \leq b$, then for all strategies~$\tau$ for Player~$1$ we have~$\cost(\tau) \leq b$.
Dually, if Player~$1$ has a strategy~$\tau$ with $\cost(\tau) > b$, then for all strategies~$\sigma$ for Player~$0$ we have~$\cost(\sigma) > b$.
We say that a quantitative game is \textit{determined} if for each~$b \in \nats$, either Player~$0$ has a strategy with cost at most~$b$, or Player~$1$ has a strategy with cost strictly greater than~$b$.

We say that~$b \in \nats$ is a \textit{cap} of a quantitative game~$\game$ if Player~$0$ winning~$\game$ implies that she has a strategy with cost at most~$b$.
A cap~$b$ for a game~$\game$ is \textit{tight} if it is minimal.

%% file: quantitative-reductions.tex
\section{Quantitative Reductions}
\label{sec:quantitative-reductions}

\newcommand{\capfunc}{\text{cap}}

Before defining quantitative reductions, we first recall the definition of qualitative ones.
To this end, let~$\game = (\arena, \wincond)$ and~$\game' = (\arena', \wincond')$ be qualitative games.
We say that~$\game$ is reducible to~$\game'$ via the memory structure~$\mem$ for~$\arena$ if~$\arena' = \arena \times \mem$ and if~$\rho \in \wincond$ if and only if~$\ext(\rho) \in \wincond'$.
Then, Player~$0$ wins~$\game$ if and only if she wins~$\game'$.
Moreover, if~$\sigma'$ is a winning strategy for Player~$0$ in~$\game'$ that is implemented by~$\mem'$, then a winning strategy for her in~$\game$ is implemented by~$\mem \times \mem'$.

We now define quantitative reductions as an analogous technique for the study of quantitative games and show that they exhibit the same properties as qualitative reductions.
Intuitively, given two quantitative games~$\game$ and~$\game'$, we aim to say that~$\game$ is reducible to~$\game'$ if plays in one game can be translated into plays in the other, retaining their cost along the transformation.
In fact, two such associated plays do not need to carry identical cost, but it suffices that the order on plays induced by their cost is retained.

To capture such order-retaining transformations of cost functions, we introduce~$b$-correction functions.
Let~$b \in \nats_\infty$.
A function~$f\colon \nats_\infty \rightarrow \nats_\infty$ is a $b$-correction function if
\begin{itemize}
	\item for all~$b'_1 < b'_2 < b$ we have~$f(b'_1) < f(b'_2)$,
	\item for all~$b' < b$ we have~$f(b') < f(b)$, and
	\item for all~$b' \geq b$ we have~$f(b') \geq f(b)$.
\end{itemize}
Thus, intuitively, a~$b$-correction function is strictly monotonic up to, but not including,~$b$, and~$f(b)$ is a lower bound for all values~$f(b')$ with~$b' \geq b$.

For $b = \infty$ these requirements degenerate to demanding that~$f$ is strictly monotonic, which in turn implies~$f(\infty) = \infty$ and~$f(b) \neq \infty$ for all~$b \neq \infty$.
Dually, if $b = 0$, we only require that~$f(0)$ bounds the values of~$f(b)$ from below.
As an example, for each~$b \in \nats_\infty$ we define the function~$\capfunc_b$ as follows:
\[
\capfunc_b(b') = \begin{cases}
	\min\set{b, b'} &\text{if }b' \neq \infty \text{ and} \\
	\infty &\text{otherwise \enspace .}
\end{cases}
\]
Then, the function~$\capfunc_b$ is a~$b$-correction function.

Leveraging the notion of~$b$-correction functions, we are now able to define quantitative reductions.
Let~$\game = (\arena,\cost)$ and~$\game' = (\arena',\cost')$ be quantitative games, let~$\mem$ be some memory structure for~$\arena$, let~$b \in \nats_\infty$, and let~$f\colon \nats_\infty \rightarrow \nats_\infty$ be some function.
We say that~$\game$ is \emph{$b$-reducible to~$\game'$ via~$\mem$ and~$f$} if
\begin{itemize}
	\item $\arena' = \arena \times \mem$,
	\item $f$ is a $b$-correction function,
	\item $\cost'(\ext(\rho)) = f(\cost(\rho))$ for all plays~$\rho$ of~$\arena$ with~$\cost(\rho) < b$, and
	\item $\cost'(\ext(\rho)) \geq f(b)$ for all plays~$\rho$ of~$\arena$ with~$\cost(\rho) \geq b$.
\end{itemize}
We write~$\game \reducesto^b_{\mem,f} \game'$ in this case.
Moreover, we use~$\capfunc_b$ as a ``default'' function: If~$f = \capfunc_b$, we omit stating~$f$ explicitly and write~$\game \reducesto^b_\mem \game'$.
The penultimate condition implies that for each play~$\ext(\rho)$ in~$\arena'$ with~$\cost'(\ext(\rho)) \leq f(b)$ there exists some~$b'$ such that~$\cost'(\ext(\rho)) = f(b')$.
Clearly, quantitative reductions are downward-closed with respect to the parameter~$b$:
If~$\game \reducesto^b_{\mem, f} \game'$ for some~$b \in \nats_\infty$, then, for all~$b' \leq b$, we have~$\game \reducesto^{b'}_{\mem, f} \game'$.

Qualitative reductions retain whether or not a play is winning.
In quantitative games, the notion of winning is refined to the notion of cost of a play.
Hence, our next aim is to show that quantitative reductions indeed retain the costs of strategies.
To this end, we first demonstrate that correction functions tie the cost of plays in~$\game'$ to that of plays in~$\game$.

\begin{lemma}
\label{lem:cost-mapping-properties}
Let~$\game$ and~$\game'$ be quantitative games such that~$\game \reducesto^{b}_{\mem, f} \game'$, for some~$b \in \nats_\infty$, some memory structure $\mem$, and some~$b$-correction function~$f$.
All of the following hold true for all~$b' \in \nats$ and all plays~$\rho$ in~$\game$:
\begin{enumerate}
	\item\label{lem:cost-mapping-properties:1} If $b' < b$ and~$\cost'(\ext(\rho)) < f(b')$, then $\cost(\rho) < b'$.
	\item\label{lem:cost-mapping-properties:2} If $b' < b$ and~$\cost'(\ext(\rho)) = f(b')$, then $\cost(\rho) = b'$.
	\item\label{lem:cost-mapping-properties:4} If $b' < b$ and~$\cost'(\ext(\rho)) > f(b')$, then $\cost(\rho) > b'$.
	\item\label{lem:cost-mapping-properties:3} If $\cost'(\ext(\rho)) \geq f(b)$, then $\cost(\rho) \geq b$.
\end{enumerate}
\end{lemma}

\begin{proof}
\ref{lem:cost-mapping-properties:1}) Let~$b' < b$ and let~$\rho$ be such that $\cost'(\ext(\rho)) < f(b')$.
Towards a contradiction assume $\cost(\rho) = b'' \geq b'$.
We have $f(b'') = f(\cost(\rho)) = \cost'(\ext(\rho))$.
If $b'' < b$, then we obtain $f(b') \leq f(b'')$, which implies $f(b') \leq \cost'(\ext(\rho))$, contradicting $\cost'(\ext(\rho)) < f(b')$.
If, however, $b'' \geq b$, then $\cost'(\ext(\rho)) = f(b'') \geq f(b) > f(b')$, again contradicting $\cost'(\ext(\rho)) < f(b')$.

\ref{lem:cost-mapping-properties:2}) Let~$b' < b$ and let~$\rho$ be such that $\cost'(\ext(\rho)) = f(b')$.
Towards a contradiction assume $\cost(\rho) = b'' \neq b'$.
We again have $f(b'') = \cost'(\ext(\rho))$.
First assume $b'' < b'$.
Then we have $b'' < b' < b$, which implies $f(b'') < f(b')$, contradicting $\cost'(\ext(\rho)) = f(b')$.
If $b' < b'' < b$, we obtain the contradiction $\cost'(\ext(\rho)) > f(b')$ analogously.
Finally, if $b \leq b''$, then $\cost'(\ext(\rho)) = f(b'') \geq f(b) > f(b')$, which again contradicts $\cost'(\ext(\rho)) = f(b')$.

\ref{lem:cost-mapping-properties:4}) Let~$b' < b$ and let~$\rho$ be such that~$\cost'(\ext(\rho)) > f(b')$.
Towards a contradiction, assume~$\cost(\rho) \leq b'$.
We then obtain~$f(\cost(\rho)) \leq f(b')$ due to $b' < b$ and due to~$f$ being a~$b$-correction function.
Furthermore, due to~$\cost(\rho) \leq b'$, we obtain~$\cost(\rho) < b$.
Hence, we have~$f(\cost(\rho)) = \cost'(\ext(\rho))$ due to the third property from the definition of~$b$-reducibility.
Since we furthermore have $f(\cost(\rho)) \leq f(b')$ as argued above.
This, in turn, directly implies $\cost'(\ext(\rho)) \leq f(b')$, which contradicts the assumption $\cost'(\ext(\rho)) > f(b')$.

\ref{lem:cost-mapping-properties:3}) Let~$\rho$ be such that $\cost'(\ext(\rho)) \geq f(b)$.
Towards a contradiction assume $\cost(\rho) = b' < b$.
We again have $f(b') = \cost'(\ext(\rho))$.
However, we obtain $f(b') < f(b)$ due to~$f$ being a $b$-correction function.
This contradicts $\cost'(\ext(\rho)) = f(b') \geq f(b)$.
\qed
\end{proof}

These properties of correction functions when used in quantitative reductions enable us to state and prove the main result of this section, which establishes quantitative reductions as the quantitative counterpart of qualitative reductions:
If~$\game \reducesto^{b+1}_{\mem, f} \game'$, then all plays of cost at most~$b$ in~$\game$ are \myquot{tracked} precisely in~$\game'$.
Hence, as long as the cost of a strategy in~$\game$ is at most~$b$, it is possible to construct a strategy in~$\game'$ with cost at most~$f(b)$.
This holds true for both players.

If a strategy has cost greater than~$b$, however, we do not have a direct correspondence between costs of plays in~$\game$ and~$\game'$ anymore.
If, however,~$b$ additionally is a cap of~$\game$, and if~$\game$ is determined, then we can still show that Player~$1$ has a strategy of infinite cost in~$\game$ if he has a strategy of cost greater than~$f(b)$ in~$\game'$.

\begin{theorem}
\label{thm:capped-reduction} 
Let~$\game$ and~$\game'$ be determined quantitative games with $\game \reducesto^{b+1}_{\mem, f} \game'$ for some~$b$, $\mem$, and~$f$, where~$b \in \nats$ is a cap of~$\game$.
\begin{enumerate}
	\item\label{thm:capped-reduction:exact} Let $b' < b+1$. Player~$i$ has a strategy~$\sigma'$ in~$\game'$ with $\cost'(\sigma') = f(b')$ if and only if they have a strategy~$\sigma$ in~$\game$ with $\cost(\sigma) = b'$.
	\item\label{thm:capped-reduction:cap} If Player~$1$ has a strategy~$\tau'$ in~$\game'$ with~$\cost'(\tau') \geq f(b+1)$, then he has a strategy~$\tau$ in~$\game$ with~$\cost(\tau) = \infty$.
\end{enumerate}
\end{theorem}

\begin{proof}
\ref{thm:capped-reduction:exact}) %
We first show the direction from left to right.
To this end, let~$\sigma'$ be a strategy for Player~$i$ in~$\game'$ with~$\cost'(\sigma') = f(b')$ for some~$b' \leq b$.
We define the strategy~$\sigma$ for Player~$i$ in~$\game$ for all play prefixes~$\pi$ ending in a vertex in~$V_i$ via~$\sigma(\pi) = v$ if~$\sigma'(\ext(\pi)) = (v, m)$ for some~$m \in M$ and claim~$\cost(\sigma) = b'$.
To this end, we first show~$\cost(\sigma) \leq b'$ for the case $i=0$ and~$\cost(\sigma) \geq b'$ for the case~$i=1$.

Let~$\rho$ be an infinite play consistent with~$\sigma$.
A straightforward induction shows that $\rho' = \ext(\rho)$ is consistent with~$\sigma'$.
If~$i = 0$, i.e., if~$\sigma$ is a strategy for Player~$0$, then $\cost'(\rho') = \cost'(\ext(\rho)) \leq \cost'(\sigma') = f(b')$.
Since we furthermore have~$b' < b+1$ by assumption and since~$f$ is a~$b+1$-correction function, we obtain $\cost(\rho) \leq b'$, due to Lemma~\ref{lem:cost-mapping-properties}.\ref{lem:cost-mapping-properties:1} and Lemma~\ref{lem:cost-mapping-properties}.\ref{lem:cost-mapping-properties:2}.
Since we picked~$\rho$ arbitrarily from the plays consistent with~$\sigma$, this in turn implies~$\cost(\sigma) \leq b'$.
If, however~$i = 1$, we directly obtain~$\cost'(\ext(\rho)) \geq f(b')$ due to $\cost(\sigma) \geq f(b')$.
Hence, we furthermore obtain $\cost(\rho) \geq b'$ due to Lemma~\ref{lem:cost-mapping-properties}.\ref{lem:cost-mapping-properties:2} and Lemma~\ref{lem:cost-mapping-properties}.\ref{lem:cost-mapping-properties:4}.
Since we again picked~$\rho$ arbitrarily from the plays consistent with~$\sigma$, this in turn implies~$\cost(\sigma) \geq b'$.
It remains to show~$\cost(\sigma) \geq b'$ for the case~$i=0$ and~$\cost(\sigma) \leq b'$ for the case~$i = 1$.

First, since~$b' < b + 1 < \infty$, we obtain~$f(b') < \infty$:
If $f(b') = \infty$, we obtain $f(b'+1) = \infty$, which contradicts strict monotonicity of~$f$ up to and including~$b+1 \geq b'+1$.
Since~$\cost'(\sigma') = f(b') < \infty$, there exists a play~$\rho'$ that is consistent with~$\sigma'$ such that~$\cost'(\rho') = f(b')$.

Further, let~$\rho$ be the unique play such that~$\ext(\rho) = \rho'$.
By induction, we obtain that~$\rho$ is consistent with~$\sigma$.
Additionally, we have $\cost(\rho) = b'$ due to~$\cost'(\rho') = \cost'(\ext(\rho)) = f(b')$ and Lemma~\ref{lem:cost-mapping-properties}.\ref{lem:cost-mapping-properties:2}.
Hence, $\cost(\sigma) \geq b'$ if~$i=0$, and $\cost(\sigma) \leq b'$ if~$i=1$, which concludes the direction from left to right.

In order to show the inverse direction of the statement, let~$\sigma$ be a strategy in~$\game$ with~$\cost(\sigma) = b' < b + 1$.
We define the strategy~$\sigma'$ for Player~$i$ in~$\game'$ for all play prefixes~$\ext(\pi) = (v_0,m_0)\cdots(v_j,m_j)$ ending in a vertex in~$V_i \times M$ as~$\sigma'(\ext(\pi)) = (v, \update(m_j, v))$ if~$\sigma(\pi) = v$ and claim~$\cost'(\sigma') = f(b')$.

Let~$\ext(\rho)$ be a play consistent with~$\sigma'$.
A straightforward induction yields that~$\rho$ is consistent with~$\sigma$, hence, if~$i=0$, then $\cost(\rho) \leq b'$ and thus $\cost'(\ext(\rho)) \leq f(b')$ due to~$b' < b+1$.
Dually, if~$i=1$, then $\cost(\rho) \geq b'$ and~$\cost'(\ext(\rho)) \geq f(b')$.
Hence, we obtain~$\cost'(\sigma') \leq f(b')$ if~$i=0$ as well as~$\cost'(\sigma') \geq f(b')$ if~$i=1$.
It remains to show~$\cost'(\sigma') \geq f(b')$ and~$\cost'(\sigma') \leq f(b')$ in the former and latter case, respectively.

Now let~$\rho$ be a play consistent with~$\sigma$ such that~$\cost(\rho) = b'$.
Since~$b' < \infty$, such a play exists.
Via another straightforward induction we obtain that~$\ext(\rho)$ is consistent with~$\sigma'$.
As $\cost'(\ext(\rho)) = f(b')$, we furthermore obtain $\cost'(\sigma') \geq f(b')$ if~$i=0$ and~$\cost'(\sigma') \leq f(b')$ if~$i=1$, which concludes the proof of this statement.

\ref{thm:capped-reduction:cap}) Let~$\tau'$ be a strategy for Player~$1$ in~$\game'$ with $\cost'(\tau') \geq f(b + 1)$.
We define the strategy~$\tau$ for Player~$1$ in~$\game$ via~$\tau(\pi) = v$ if $\tau'(\ext(\pi)) = (v,m)$ for all play prefixes~$\pi$ in~$\game$.
Let~$\rho$ be a play consistent with~$\tau$ and define~$\rho' = \ext(\rho)$.
A straightforward induction yields that~$\rho'$ is consistent with~$\tau'$.
Since~$\cost'(\tau') \geq f(b+1)$, we obtain~$\cost'(\rho') \geq f(b+1)$.
Then, we obtain~$\cost(\rho) \geq b+1$ due to Lemma~\ref{lem:cost-mapping-properties}.\ref{lem:cost-mapping-properties:3}.
Since we picked~$\rho$ arbitrarily from the plays consistent with~$\tau$, we directly obtain~$\cost(\tau) \geq b+1$.
Since~$b$ is a cap of~$\game$ and due to determinacy of~$\game$, this implies that there exists a strategy~$\tau''$ for Player~$1$ in~$\game$ such that~$\cost(\tau'') = \infty$.
\qed
\end{proof}

We proved Theorem~\ref{thm:capped-reduction} by constructing optimal strategies for Player~$0$ in~$\game$ from optimal strategies for her in~$\game'$.
These strategies use the set of all play prefixes of~$\game'$ as memory states and may thus be of infinite size.
If Player~$0$ can achieve a certain cost in~$\game'$ using a finite-state strategy, however, then she can achieve the corresponding cost in~$\game$ with a finite-state strategy as well.

\begin{theorem}
	\label{thm:reductions:strategy-lift}
	Let~$\game$ and~$\game'$ be quantitative games such that~$\game \reducesto^b_{\mem_1,f} \game'$ for some~$b$,~$\mem_1$, and~$f$ and let either $b' < b$ or $b' = b = \infty$.
	If Player~$i$ has a finite-state strategy~$\sigma'$ with~$\cost'(\sigma') = f(b')$ in~$\game'$ that is implemented by~$\mem_2$, then she has a finite-state strategy~$\sigma$ with~$\cost(\sigma) = b'$ in~$\game$ that is implemented by~$\mem_1 \times \mem_2$.
\end{theorem}

\begin{proof}
Let~$\game = (\arena, \cost)$,~$\game' = (\arena', \cost')$, $\mem_1 = (M_1, m^1_\initmark, \update_1)$, and~$\mem_2 = (M_2, m^2_\initmark, \update_2)$ such that~$\sigma'$ is implemented by~$\mem_2$ with the next-move function~$\nxt'\colon (V \times M_1) \times M_2  \rightarrow (V \times M_1)$.
We define~$\nxt(v, (m_1, m_2)) = v^*$ if~$\nxt'((v, m_1), m_2) = (v^*, \update_1(m_1, v^*))$.
We moreover define~$\sigma$ as the strategy that is implemented by~$\mem_1 \times \mem_2$ with the next-move function~$\nxt$.

Let~$\rho = v_0v_1v_2\cdots$ be a play consistent with~$\sigma$, let
\[
	\ext_{\mem_1 \times \mem_2}(\rho) = (v_0, m_1^0, m_2^0)(v_1, m_1^1, m_2^1)(v_2, m_1^2, m_2^2)\cdots
\]
be its extension with respect to~$\mem_1 \times \mem_2$, and let~$j \in \nats$ be such that~$v_j \in V_i$.
We obtain $v_{j+1} = \sigma(v_0\cdots v_j) = \nxt(v_j, (m_1^j, m_2^j))$.
Due to the definition of $\nxt$, this implies $\nxt'((v_j, m_1^j), m_2^j) = (v_{j+1}, m_1^{j+1})$, where $m_1^{j+1} = \update_1(m_1^j, v_{j+1})$ due to the construction of~$\arena \times \mem_1$.
Hence,~$\ext_{\mem_1}(\rho)$ is consistent with~$\sigma'$, i.e.,~$\cost'(\ext_{\mem_1}(\rho)) \leq f(b')$, which in turn implies $\cost(\rho) \leq b'$ for~$i = 0$ due to Lemma~\ref{lem:cost-mapping-properties}.\ref{lem:cost-mapping-properties:1} and Lemma~\ref{lem:cost-mapping-properties}.\ref{lem:cost-mapping-properties:2}, and~$\cost'(\ext_{\mem_1}(\rho)) \geq f(b')$ and~$\cost(\rho) \geq b'$ for~$i = 1$.

Due to similar reasoning, for each play~$\ext_{\mem_1}(\rho)$ consistent with~$\sigma'$, the play~$\rho$ is consistent with~$\sigma$.
If~$i = 1$ or $b' < \infty$, this concludes the proof.
If, however,~$i = 0$ and~$b' = \infty$, then we furthermore obtain that~$b = \infty$ and that~$f$ is a strictly monotonic function with~$f(\infty) = \infty$.
Hence, if there exists a play~$\ext_{\mem_1}(\rho)$ consistent with~$\sigma'$ with~$\cost'(\ext_{\mem_1}(\rho)) = \infty$, then~$\cost(\rho) = \infty$ and, hence,~$\cost(\sigma) = \infty$.
If, however, the costs of the plays consistent with~$\sigma'$ diverges, then the cost of the plays consistent with~$\sigma$ diverges as well and we obtain~$\cost(\sigma) = \infty$.
\qed
\end{proof}

Theorem~\ref{thm:capped-reduction} and Theorem~\ref{thm:reductions:strategy-lift} show that quantitative reductions indeed exhibit properties analogous to those of qualitative reductions in the quantitative setting.
Recall that in addition to retaining winning plays and to allowing the implementation of finite-state strategies, qualitative reductions furthermore are transitive:
If~$\game$,~$\game'$, and~$\game''$ are qualitative games such that~$\game \reducesto \game'$, and such that~$\game' \reducesto \game''$, then~$\game \reducesto \game''$.
We now show that quantitative reductions are transitive as well.

\begin{theorem}
\label{thm:quantitative-reductions:transitivity}
Let~$\game_1,\game_2,\game_3$ be quantitative games such that $\game_1 \reducesto^{b_1}_{\mem_1,f_1} \game_2$ and $\game_2 \reducesto^{b_2}_{\mem_2,f_2} \game_3$ for some~$b_1, b_2 \in \nats_\infty$, some memory structures $\mem_1, \mem_2$, and some~$b_1$- and~$b_2$-correction functions $f_1$ and $f_2$, respectively.

Then, we have~$\game_1 \reducesto^{b}_{\mem, f} \game_3$, where~$\mem = \mem_1 \times \mem_2$, $f = f_2 \circ f_1$, and~$b = b_1$ if~$b_2 \geq f_1(b_1)$ and~$b = \max\set{b' \mid f_1(b') \leq b_2}$ otherwise.
\end{theorem}

\begin{proof}
	For each~$j \in \set{1,2,3}$, let~$\game_j = (\arena_j, \cost_j)$.
	Recall that in order to show $\game_1 \reducesto^{b}_{\mem, f} \game_3$ we have to show that
	\begin{itemize}
		\item $\arena_3 = \arena_1 \times \mem_1 \times \mem_2$,
		\item $f_2 \circ f_1$ is a $b$-reduction function,
		\item for all plays~$\rho$ of~$\arena_1$ with~$\cost_1(\rho) < b$ we have $\cost_3(\ext_{\mem_1 \times \mem_2}(\rho)) = (f_2 \circ f_1)(\cost_1(\rho))$, and
		\item for all plays~$\rho$ of~$\arena_1$ with~$\cost_1(\rho) \geq b$ we have $\cost_3(\ext_{\mem_1 \times \mem_2}(\rho)) \geq (f_2 \circ f_1)(b)$.
	\end{itemize}
	We show all these items individually.	
	
	Clearly, we have
\[
	\arena_3 = \arena_2 \times \mem_2 = \arena_1 \times \mem_1 \times \mem_2 \enspace .
\]

	We now show that~$f_2 \circ f_1$ is a $b$-correction function, for~$b$ defined as in the statement of the theorem.
	Recall that to this end we have to show that
	\begin{itemize}
		\item for all $x < x' < b$ we have $(f_2 \circ f_1)(x) < (f_2 \circ f_1)(x')$,
		\item for all $x < b$ we have $(f_2 \circ f_1)(x) < (f_2 \circ f_1)(b)$, and
		\item for all $x \geq b$ we have $(f_2 \circ f_1)(x) \geq (f_2 \circ f_1)(b)$.
	\end{itemize}
	Furthermore recall that we defined
	\[
	b = \begin{cases}
		b_1 & \text{if $b_2 \geq f_1(b_1)$ and} \\
		\max\set{b' \mid f_1(b') \leq b_2} & \text{otherwise} \enspace .
	\end{cases}
	\]
	We treat both cases of this definition separately.

	First, assume~$b_2 \geq f_1(b_1)$.
	In this case we have $b = b_1$.
	We show the three items of the definition of a $b$-correction function independently.
	
	\begin{itemize}
	\item%
		First, pick~$x$ and~$x'$ such that $x < x' < b = b_1$.
		We show $(f_2 \circ f_1)(x) < (f_2 \circ f_1)(x')$.
		Since~$f_1$ is a~$b_1$-correction function, we obtain~$f_1(x) < f_1(x')$, $f_1(x) < f_1(b_1)$, and $f_1(x') < f_1(b_1)$.
		Since, furthermore,~$f_2$ is a~$b_2$-correction function and as $f_1(b_1) \leq b_2$ by assumption, we moreover obtain $(f_2 \circ f_1)(x) < (f_2 \circ f_1)(x')$.

	\item%
		Now pick some~$x$ such that~$x < b = b_1$.
		We show $(f_2 \circ f_1)(x) < (f_2 \circ f_1)(b)$.
		Since $f_1$ is a $b_1$-correction function, we directly obtain~$f_1(x) < f_1(b_1)$.
		Since $f_1(b_1) \leq b_2$ by assumption and since~$f_2$ is a $b_2$-correction function, this directly yields $(f_2 \circ f_1)(x) < (f_2 \circ f_1)(b_1) = (f_2 \circ f_1)(b)$.
	
	\item%
		Finally, pick some~$x$ such that~$x \geq b = b_1$.
		Then $f_1(x) \geq f_1(b_1)$, since, again,~$f_1$ is a $b_1$-correction function.
		If $f_1(x) < b_2$, then $(f_2 \circ f_1)(x) \geq (f_2 \circ f_1)(b_1)$.
		If, however, $f_1(x) \geq b_2$, then $(f_2 \circ f_1)(x) \geq f_2(b_2) \geq (f_2 \circ f_1)(b_1)$, where the latter inequality follows from the assumption $b_2 \geq f_1(b_1)$.
		This concludes this part of the proof.
	\end{itemize}
	
	Now assume $b_2 < f_1(b_1)$ and let $b$ be maximal such that $f_1(b) \leq b_2$.
	We first argue that we have $b \leq b_1$ in this case.
	Towards a contradiction assume $ b_1 < b$.
	This implies $f_1(b_1) \leq f_1(b)$ due to~$f_1$ being a $b_1$-correction function.
	However, we have $f_1(b) \leq b_2 < f_1(b_1)$, where the former inequality results from the definition of~$b$, while the latter one is due to our initial assumption.
	This directly contradicts $f_1(b_1) \leq f_1(b)$, hence we obtain $b \leq b_1$.
	
	We again show that~$f_2 \circ f_1$ is a~$b$-correction function by showing the three items of the definition independently.
	
	\begin{itemize}
	\item%
		First, pick~$x$ and~$x'$ such that~$x < x' < b$.
		Since $b \leq b_1$ we obtain $f_1(x) < f_1(x') < f_1(b)$ due to $f_1$ being a $b_1$-correction function.
		As $f_1(b) \leq b_2$ due to the definition of~$b$, we directly obtain $(f_2 \circ f_1)(x) < (f_2 \circ f_1)(x')$.

	\item%
		Now, pick~$x$ such that $x < b$.
		We obtain $(f_2 \circ f_1)(x) < (f_2 \circ f_1)(b)$ via reasoning analogous to the previous case.

	\item%
		Finally, pick~$x$ such that $b \leq x$.
		We show $(f_2 \circ f_1)(b) \leq (f_2 \circ f_1)(x)$.
		To this end, we first observe that we have $f_1(b) \leq f_1(x)$ by leveraging different properties of~$f_1$ being a $b_1$-correction function, depending on whether we have $b \leq x < b_1$ or $b \leq b_1 \leq x$.
		The case $b_1 < b \leq x$ is excluded due to $b \leq b_1$ as argued above.
		We furthermore obtain $(f_2 \circ f_1)(b) \leq (f_2 \circ f_1)(x)$ by similar reasoning, using different properties of~$f_2$ being a $b_2$-correction function, depending on whether we have $f_1(b) \leq b_2 \leq f_1(x)$ or $f_1(b) \leq f_1(x) < b_2$.
		Again, the case $b_2 < f_1(b) \leq f_1(x)$ is excluded due to the definition of~$b$.
	\end{itemize}
	
	Thus, we have shown that $f_2 \circ f_1$ is indeed a $b$-correction function.
	It remains to prove the latter two conditions from the definition of $\game_1 \reducesto^{b}_{\mem, f} \game_3$.
	We first aim to show that for all plays $\rho$ of~$\arena_1$ with~$\cost_1(\rho) < b$ we have $\cost_3(\ext_{\mem_1 \times \mem_2}(\rho)) = (f_2 \circ f_1)(\cost_1(\rho))$.
	To this end, let~$\rho$ be a play of~$\arena_1$ with~$\cost_1(\rho) < b$.
	For the sake of consistency and readability, we define $\rho_1 = \rho$, $\rho_2 = \ext_{\mem_1}(\rho_1)$, and~$\rho_3 = \ext_{\mem_2}(\rho_2) = \ext_{\mem_1 \times \mem_2} (\rho_1)$.
	We again treat both cases of the definition of~$b$ separately.
	
	We again first consider the case that~$b_2 \geq f_1(b_1)$.
	We then directly obtain $b = b_1$ due to the definition of~$b$.
	This yields $\cost_1(\rho_1) < b_1$, which in turn implies $f_1(\cost_1(\rho_1)) = \cost_2(\rho_2)$ due to~$\game_1 \reducesto^{b_1}_{\mem_1, f_1} \game_2$.
	Furthermore, since $\cost_1(\rho_1) < b_1$ and since~$f_1$ is a~$b_1$-correction function we have $f_1(\cost_1(\rho_1)) < f_1(b_1)$,	which directly yields $\cost_2(\rho_2) < f_1(b_1)$ via the equation above.
	Since $f_1(b_1) \leq b_2$, we further obtain $\cost_2(\rho_2) < b_2$.
	Due to~$\game_2 \reducesto^{b_2}_{ \mem_2, f_2} \game_3$ this then implies $\cost_3(\rho_3) = f_2(\cost_2(\rho_2))$.
	By again applying $f_1(\cost_1(\rho_1)) = \cost_2(\rho_2)$, we obtain the desired result of $\cost_3(\rho_3) = (f_2 \circ f_1)(\cost_1(\rho_1))$.
	
	Now consider the case that $b_2 < f_1(b_1)$.
	In this case, we have $b = \max\set{b' \mid f_1(b') \leq b_2}$ by definition of~$b$.
	As argued above, we have $b \leq b_1$ in this case, which directly implies $\cost_1(\rho_1) < b_1$.
	Due to $\game_1 \reducesto^{b_1}_{\mem_1, f_1} \game_2$ this yields $f_1(\cost_1(\rho_1)) = \cost_2(\rho_2)$.
	Moreover, we again have $f_1(b) \leq b_2$ by definition of~$b$.
	Since $\game_1 \reducesto^{b_1}_{\mem_1, f_1} \game_2$ and since $\cost_1(\rho_1) < b$, the latter due to our choice of~$\rho_1$, we also obtain $f_1(\cost_1(\rho_1)) < f_1(b)$.
	Moreover, we have $f_1(b) \leq b_2$ due to the definition of~$b$.
	Hence, we obtain $f_1(\cost_1(\rho_1)) < b_2$.
	By applying~$\game_2 \reducesto^{b_2}_{\mem_2, f_2} \game_3$ analogously to the previous case, we obtain the desired result of $\cost_3(\rho_3) = (f_2 \circ f_1)(\cost_1(\rho_1))$.
	
	It remains to show the final item in the definition of $b$-reducibility, i.e., that we indeed have $\cost_3(\ext_{\mem_1 \times \mem_2}(\rho)) \geq (f_2 \circ f_1)(b)$ for all plays~$\rho$ of~$\arena_1$ with~$\cost_1(\rho) \geq b$.
	To this end, let~$\rho$ be a play of~$\arena_1$ with~$\cost_1(\rho) \geq b$.
	For the sake of brevity and consistency we again define $\rho_1 = \rho$, $\rho_2 = \ext_{\mem_1}(\rho)$, and $\rho_3 = \ext_{\mem_2}(\rho_2) = \ext_{\mem_1 \times \mem_2}(\rho_1)$.
	We again treat the two cases from the definition of~$b$ separately.
	
	First, assume~$b_2 \geq f_1(b_1)$.
	We directly obtain~$b = b_1$ by definition of~$b$, which in turn implies~$\cost_2(\rho_2) \geq f_1(b_1)$ due to $\game_1 \reducesto^{b_1}_{\mem_1, f_1} \game_2$.
	If $b_2 > \cost_2(\rho_2)$, then we directly obtain $f_2(\cost_2(\rho_2)) \geq (f_2 \circ f_1)(b_1)$ since~$f_2$ is a $b_2$-correction function.
	Since we have $\game_2 \reducesto^{b_2}_{\mem_2, f_2} \game_3$, we moreover obtain $f_2(\cost_2(\rho_2)) =  \cost_3(\rho_3)$, which in turn yields $\cost_3(\rho_3) \geq (f_2 \circ f_1) (b_1) = (f_2 \circ f_1) (b)$.
	If, however, $\cost_2(\rho_2) \geq b_2$, then we obtain $\cost_3(\rho_3) \geq f_2(b_2) \geq (f_2 \circ f_1)(b_1) = (f_2 \circ f_1)(b)$, where the former inequality is implied by~$\game_2 \reducesto^{b_2}_{\mem_2, f_2} \game_3$ while we obtain the latter due to $f_2$ being a $b_2$-correction function and due to $b_2 \geq f_1(b_1)$.
	
	Now consider the case $b_2 < f_1(b_1)$.
	Here, we again distinguish two sub-cases.
	If $\cost_1(\rho_1) \geq b_1$, then we obtain $\cost_2(\rho_2) \geq f_1(b_1) > b_2 \geq f_1(b)$.
	The three inequalities result from $\game_1 \reducesto^{b_1}_{\mem_1, f_1} \game_2$, from the assumption above, and from the definition of~$b$, respectively.
	We furthermore obtain $\cost_3(\rho_3) \geq f_2(b_2) \geq (f_2 \circ f_1)(b)$, where the former inequality is due to $\game_2 \reducesto^{b_2}_{\mem_2, f_2} \game_3$ and $\cost_2(\rho_2) > b_2$, while the latter one results from~$f_2$ being a $b_2$-correction function and from $b_2 \geq f_1(b)$.
	
	It remains to consider the case $b_1 > \cost_1(\rho_1)$.
	In this case, we directly obtain $\cost_2(\rho_2) = f_1(\cost_1(\rho_1)) \geq f_1(b)$.
	The former equality is due to $\game_1 \reducesto^{b_1}_{\mem_1, f_1} \game_2$, while the latter results from the assumption $\cost_1(\rho_1) \geq b$ and the fact that~$f_1$ is a $b_1$-correction function.
	We again distinguish two cases based on the relation of~$b_2$ to~$f_1(b)$ and~$\cost_2(\rho_2)$: Either we have $\cost_2(\rho_2) \geq b_2 \geq f_1(b)$, or we have $b_2> \cost_2(\rho_2) \geq f_1(b)$.
	In the former case, we obtain $\cost_3(\rho_3) \geq f_2(b_2) \geq (f_2 \circ f_1)(b)$, where the former and latter inequality are due to $\game_2 \reducesto^{b_2}_{\mem_2, f_2} \game_3$ and due to~$f_2$ being a $b_2$-reduction function, respectively.
	In the latter case, we similarly obtain $f_2(\cost_2(\rho_2)) = \cost_3(\rho_3) \geq (f_2 \circ f_1)(b)$, where the former equality again results from $\game_2 \reducesto^{b_2}_{\game_2, f_2} \game_3$, while the latter inequality is again due to $f_2$ being a $b_2$-correction function.
\qed
\end{proof}

By using quantitative reductions we are able to structure the space of quantitative games similarly to that of qualitative games.
While for qualitative games, there exist direct solutions to a number of well-studied winning conditions, for quantitative games, no such direct solutions for the threshold problems exist to the best of our knowledge.
Instead, the threshold problem for quantitative games is usually solved by reducing the quantitative game to a qualitative game for a fixed bound~$b$.

Hence, there does not yet exist a ``foundation'' of the space of quantitative winning conditions analogous to that of the space of qualitative winning conditions, i.e., there is no canonical simple class of quantitative games that provides a natural target for quantitative reductions.
In the following section, we provide such a foundation in the form of vertex-ranked games.

%% file: vertex-ranked-games.tex
\section{Vertex-Ranked Games}
\label{sec:vertex-ranked-games}

\newcommand{\suprankgame}{{\textsc{Rank}^{\sup}}}
\newcommand{\limsuprankgame}{{\textsc{Rank}^{\lim}}}
\newcommand{\rankgame}{{\textsc{Rank}^{X}}}

We introduce two very simple kinds of quantitative games, which we call vertex-ranked games.
In such games, the cost of a play is determined solely by a qualitative winning condition and a ranking of the vertices of the arena by natural numbers.
We show that solving the threshold problem for either kind of game is possible with only a polynomial overhead over solving the underlying qualitative game.

Furthermore, we show that the memory structures implementing winning strategies for either player only incur a polynomial overhead in comparison to the memory structures implementing winning strategies for the underlying conditions.
Finally, we briefly discuss the optimization problem for such games, i.e., the problem of determining the minimal~$b$ such that Player~$0$ has a strategy of cost at most~$b$ in such a game.
We argue that determining such~$b$ incurs only a polynomial overhead over solving the underlying qualitative game.

Let~$\arena$ be an arena with vertex set~$V$, let~$\wincond \subseteq (V')^\omega$ be a qualitative winning condition, and let~$\rank\colon V \rightarrow \nats$ be a ranking function on vertices.
We define the quantitative vertex-ranked~$\sup$-condition
\[
	\suprankgame(\wincond, \rank)\colon v_0v_1v_2 \cdots \mapsto \begin{cases}
	 	\sup_{j \rightarrow \infty} \rank(v_j) & \text{if } v_0v_1v_2\cdots \in \wincond \text{ and} \\
	 	\infty & \text{otherwise} \enspace , \end{cases}
\]
as well as its prefix-independent version, the vertex-ranked~$\limsup$-condition
\[
	\limsuprankgame(\wincond, \rank)\colon v_0v_1v_2\cdots \mapsto \begin{cases}
	 	\limsup_{j \rightarrow \infty} \rank(v_j) & \text{if } v_0v_1v_2\cdots \in \wincond \text{ and} \\
	 	\infty & \text{otherwise} \enspace . \end{cases}
\]

A vertex-ranked $\sup$- or~$\limsup$-game~$\game = (\arena, \rankgame(\wincond, \rank))$ with $X \in \set{\sup,\lim}$ consists of an arena~$\arena$ with vertex set~$V$, a qualitative winning condition~$\wincond$, and a vertex-ranking function~$\rank\colon V \rightarrow \nats$.

If~$\game_X = (\arena, \rankgame(\wincond, \rank))$ is a vertex-ranked~$\sup$- or~$\limsup$-game, we call the game~$(\arena, \wincond)$ the qualitative game corresponding to~$\game_X$.
Moreover, if~$\gamesup$ is a vertex-ranked~$\sup$-game, we denote the vertex-ranked~$\limsup$-game with the same arena, winning condition, and rank function by~$\gamelim$ and vice versa.
In either case, we denote the corresponding qualitative game by~$\game$.

Before showing how to solve vertex-ranked~$\sup$- and~$\limsup$-games, we argue that these games allow straightforward adaptations of qualitative winning conditions to quantitative ones.
This is witnessed by the following theorem.

\begin{theorem}
\label{thm:vertex-ranked-games:generalization:sup}
Let~$\game = (\arena, \wincond)$ and~$\game' = (\arena', \wincond')$ be qualitative games and let~$\mem$ be a memory structure such that~$\game \reducesto_\mem \game'$.
Moreover, let~$\rank$ be a ranking function on vertices of~$\game$ and let~$b$ be the maximal rank assigned to a vertex of~$\game$ by~$\rank$.
Then,~$\gamesup \reducesto^{b+1}_{\mem} \gamesup'$, where~$\gamesup = (\arena, \suprankgame(\wincond, \rank))$ and~$\gamesup' = (\arena', \suprankgame(\wincond', \rank'))$ with~$\rank'(v, m) = \rank(v)$.
\end{theorem}

\begin{proof}
Clearly, we have~$\arena' = \arena \times \mem$ due to~$\game \reducesto \game'$.
Moreover, as argued above,~$\capfunc_{b+1}$ is a~$b+1$-correction function.
Hence, it remains to show that the latter two conditions in the definition of a quantitative reduction are satisfied.

To this end, let~$\rho$ be a play in~$\arena$.
If the cost of~$\rho$ is less than~$b+1$, then we obtain $\rho \in \wincond$ by definition of~$b$ and by definition of the vertex-ranked~$\sup$-condition.
This, in turn, implies~$\ext(\rho) \in \wincond'$ due to~$\game \reducesto \game'$.
Hence, we obtain
\[
	\capfunc_{b+1}(\suprankgame(\wincond, \rank)(\rho)) = \suprankgame(\wincond, \rank)(\rho) = \suprankgame(\wincond', \rank')(\ext(\rho)) \enspace .
\]

If, however, the cost of~$\rho$ is at least~$b+1$, then we have~$\rho \notin \wincond$, again due to the definition of~$b$, which implies both~$\suprankgame(\wincond, \rank)(\rho) = \infty$ as well as~$\ext(\rho) \notin \wincond$ due to~$\game \reducesto \game'$.
Hence, we obtain
\[
	\suprankgame(\wincond', \rank')(\ext(\rho)) = \infty \geq \capfunc_{b+1}(b+1) = b+1 \enspace ,
\]
which concludes the proof of the fourth condition of quantitative reductions.
\qed
\end{proof}

Clearly, the above proof can be adapted in a very straightforward way to show the analogous result for the case of vertex-ranked~$\limsup$-games.

\begin{remark}
\label{rem:vertex-ranked-games:generalization:lim}
Let~$\game = (\arena, \wincond)$ and~$\game' = (\arena', \wincond')$ be qualitative games and let~$\mem$ be a memory structure such that~$\game \reducesto_\mem \game'$.
Moreover, let~$\rank$ be a ranking function on vertices of~$\game$ and let~$b$ be the maximal rank assigned to a vertex of~$\game$ by~$\rank$.
Then,~$\gamelim \reducesto^{b+1}_{\mem} \gamelim'$, where~$\gamelim = (\arena, \limsuprankgame(\wincond, \rank))$ and~$\gamelim' = (\arena', \limsuprankgame(\wincond', \rank'))$ with~$\rank'(v, m) = \rank(v)$.
\end{remark}

The remainder of this section is dedicated to providing bounds on the complexity of solving vertex-ranked games with respect to some given bound.
In particular, we show that vertex-ranked~$\sup$-games can be solved with only an additive linear blowup compared to the complexity of solving the corresponding qualitative games.
Vertex-ranked~$\limsup$-games, on the other hand, can be solved while incurring only a polynomial blowup compared to solving the corresponding qualitative games.

\subsection{Solving Vertex-Ranked~{$\sup$}-Games}
\label{sec:vertex-ranked-games:sup}

We begin by observing that solving vertex-ranked $\sup$-games is at least as hard as solving the underlying qualitative games, since the former subsumes the latter.
This is due to the fact that Player~$0$ has a winning strategy in~$(\arena, \wincond)$ if and only if she has a strategy with cost at most zero in~$(\arena, \suprankgame(\wincond, \rank))$, where~$\rank$ is the constant function assigning zero to every vertex.

\newcommand{\gamefam}{\mathfrak{G}}
\newcommand{\ver}{\text{Ver}}
\newcommand{\cpre}{\textsc{CPre}}
\newcommand{\supgames}[1]{{#1}^\textsc{rnk}_{\sup}}
\newcommand{\limsupgames}[1]{{#1}^\textsc{rnk}_{\lim}}

We now turn our attention to finding an upper bound for the complexity of the threshold problem for vertex-ranked~$\sup$-games.
To achieve a general treatment of such games, we first introduce some notation.
Let~$\gamefam$ be a class of qualitative games.
We define the extension of~$\gamefam$ to vertex-ranked~$\sup$-games as
\[
	\supgames{\gamefam} = \{ (\arena, \suprankgame(\wincond, \rank)) \mid (\arena, \wincond) \in \gamefam, \rank \text{ is vertex-ranking function for } \arena \} \enspace .
\]

We first show that we can use a decision procedure solving games from~$\gamefam$ to solve games from~$\supgames{\gamefam}$ with respect to a given~$b$.
To this end, we remove all vertices from which Player~$1$ can enforce a visit to a vertex of rank greater than~$b$ and proclaim that Player~$0$ wins the quantitative game with respect to~$b$ if and only if she wins the qualitative game corresponding to the resulting quantitative game.
To ensure that we are able to solve the resulting qualitative game, we assume some closure properties of~$\gamefam$.
To this end, we first introduce some notation.

Let~$\game$ be a qualitative or quantitative game with vertex set~$V$.
For each~$v \in V$, we write~$\game_v$ to denote the game~$\game$ with its initial vertex replaced by~$v$.
All other components, i.e., the structure of the arena and the cost-function, remain unchanged.
Let~$\arena = (V, V_0, V_1, E, v_\initmark)$ and~$\arena' = (V', V'_0, V'_1, E', v'_\initmark)$ be arenas.
We say that~$\arena'$ is a \emph{sub-arena} of~$\arena$ if $V' \subseteq V$, $V'_0 \subseteq V_0$, $V'_1 \subseteq V_1$, $E' \subseteq E$, and $v_\initmark = v'_\initmark$ and write~$\arena' \subarenaof \arena$ in this case.

We call a class of qualitative (or quantitative) games~$\gamefam$ \emph{proper} if
\begin{itemize}
	\item for each~$(\arena, \wincond)$ (or $(\arena, \cost)$) in $\gamefam$ and each sub-arena~$\arena' \subarenaof \arena$ the game~$(\arena', \wincond')$ (or $(\arena', \cost')$), where~$\wincond'$ (or~$\cost'$) is the restriction of~$\wincond$ (or~$\cost$) to plays from~$\arena'$, is a member of~$\gamefam$ as well, if
	\item for each game~$\game \in \gamefam$ and each vertex~$v$ of~$\game$ we have $\game_v \in \gamefam$, if
	\item all games in $\gamefam$ are determined, and if
	\item all~$\game \in \gamefam$ are finitely representable.
\end{itemize}
Intuitively, the first condition ensures that games obtained by removing vertices or edges from games in~$\gamefam$ are members of~$\gamefam$ as well, whereas the latter three conditions are very weak technical requirements.
In particular the requirement that all games included in the class must be finitely representable serves mainly to enable us to talk about the size of a game.

Using this notion of proper classes of games, we are now able to formulate the main result of this section regarding vertex-ranked~$\sup$-games.

\begin{theorem}
\label{thm:vertex-ranked-games:direct:complexity}
Let~$\gamefam$ be a proper class of qualitative games~$\game$ that can be solved in time~$t(\card{\game})$ and space~$s(\card{\game})$, where~$t$ and~$s$ are monotonic functions.

Then, the following problem can be solved in time~$\bigo(n) + t(\card{\game})$ and space $\bigo(n) + s(\card{\game})$:
\myquot{Given some game~$\gamesup \in \supgames{\gamefam}$ with~$n$ vertices and some bound~$b \in \nats$, does Player~$0$ win~$\gamesup$ with respect to~$b$?}
\end{theorem}

Intuitively, in order to prove Theorem~\ref{thm:vertex-ranked-games:direct:complexity}, we show that Player~$0$ wins~$\gamesup \in \supgames{\gamefam}$ with respect to some bound~$b$ if and only if
\begin{itemize}
	\item Player~$1$ cannot enforce a visit to vertices of rank greater than~$b$ from~$v_\initmark$, and if
	\item she is able to win the game~$\gamesup$ without visiting any vertices from which Player~$1$ is able to enforce a visit to a vertex of rank greater than~$b$.
\end{itemize}

We formalize the idea of removing vertices from which one player can enforce a visit to some set of vertices by first recalling the attractor construction.
Let~$\arena = (V, V_0, V_1, E, v_\initmark)$ be an arena with~$n$ vertices and let~$X \subseteq V$.
We define~$\att{i}{}{X} = \att{i}{n}{X}$ inductively with~$\att{i}{0}{X} = X$ and
\begin{multline*}
	\att{i}{j}{X} = \set{ v \in V_i \mid \exists v' \in \att{i}{j-1}{X}.\, (v,v') \in E} \, \cup \\
		\set{ v \in V_{1-i} \mid \forall (v, v') \in E.\, v' \in \att{i}{j-1}{X}} \cup \att{i}{j-1}{X} \enspace.
\end{multline*}
Intuitively, the $i$-attractor $\att{i}{}{X}$ is the set of all vertices from which Player~$i$ can enforce a visit to~$X$.
The set~$\att{i}{}{X}$ can be computed in linear time in~$\card{E}$ and Player~$i$ has a positional strategy~$\sigma$ such that each play starting in some vertex in $\att{i}{}{X}$ and consistent with~$\sigma$ eventually encounters some vertex from~$X$~\cite{NerodeRemmelYakhnis96}.
We call~$\sigma$ an attractor strategy towards~$X$.

We furthermore formalize the notion of removing attractors from arenas:
Let~$\arena$ be an arena with vertex set~$V$, let~$X \subseteq V$, and let $A = \att{i}{}{X}$.
If~$v_\initmark \notin A$, then we define
\[
	\arena \setminus A = (V \setminus A, V_0 \setminus A, V_1 \setminus A, \set{(v, v') \in E \mid v \notin A \text{ and } v' \notin A}, v_\initmark) \enspace ,
\]
which is again an arena.
We lift this notation to qualitative (and quantitative) games~$\game = (\arena, \wincond)$ (or $(\arena, \cost)$) by defining~$\game \setminus A = (\arena \setminus A, \wincond \cap (V \setminus A)^\omega)$ (or $(\arena \setminus A, \left.\cost\right|_{(V \setminus A)^\omega})$, where $\left.\cost\right|_{(V \setminus A)^\omega}$ denotes the restriction of~$\cost$ to the domain $(V \setminus A)^\omega$).
This restriction of the winning condition and the cost fuction to vertices of~$V$ is not strictly necessary due to our definition of~$\wincond \subseteq (V')^\omega$ and of~$\cost \colon (V')^\omega \rightarrow \nats_\infty$, but it makes the resulting objects easier to reason about.

If~$v_\initmark \in A$, however, then both~$\arena \setminus A$ and $\game \setminus A$ are undefined.
The game $\game \setminus A$ can be constructed in linear time and is of size at most~$\card{\game}$.

As a first step towards the proof of Theorem~\ref{thm:vertex-ranked-games:direct:complexity}, we show that vertex-ranked $\sup$-games can be solved by using a single attractor construction and considering the qualitative game obtained by removing the resulting attractor.

\begin{lemma}
\label{lem:vertex-ranked-games:direct:reduction}
Let~$\gamefam$ be a proper class of qualitative games, let $\gamesup = (\arena, \suprankgame(\wincond, \rank)) \in \supgames{\gamefam}$ with vertex set~$V$ and initial vertex~$v_\initmark$, and let~$b \in \nats$.

Player~$0$ has a strategy with cost at most~$b$ in~$\gamesup$ if and only if~$v_\initmark \notin A$ and if she has a winning strategy in the qualitative game~$\game' = \game \setminus A$, where $\game = (\arena, \wincond)$ and $A = \att{1}{}{\set{v \in V \mid \rank(v) > b}}$.
\end{lemma}

\begin{proof}
	Let~$X_b = \set{v \in V \mid \rank(v) > b}$.
	We first show the direction from right to left, i.e., that, if~$v_\initmark \notin A$ and if Player~$0$ wins~$\game'$, say with strategy~$\sigma'$, then she has a strategy of cost at most~$b$ in~$\gamesup$.
	To this end, define~$\arena' = \arena \setminus A$.
	Since~$\arena' \subarenaof \arena$, the strategy~$\sigma'$ is a strategy for Player~$0$ in~$\gamesup$ as well, due to Player~$0$ being able to keep the play inside~$\arena'$ using~$\sigma'$.
	Hence, each play consistent with~$\sigma'$ in~$\game'$ is consistent with~$\sigma'$ in~$\game$ as well as vice versa.
	Let~$\rho$ be a play in~$\gamesup$ consistent with~$\sigma'$.
	Since~$\sigma'$ is winning for Player~$0$ in~$\game'$, we have~$\rho \in \wincond \cap (V \setminus A)^\omega \subseteq \wincond$.
	Moreover, since~$X_b \subseteq A$, and as~$\rho$ visits only vertices occurring in~$\game'$, we obtain $\suprankgame(\wincond, \rank)(\rho) \leq b$ and thus $\cost(\sigma') \leq b$, which concludes this direction of the proof.
	
	We show the other direction via contraposition:
	To this end, first assume~$v_\initmark \in A$ and let~$\tau_A$ be an attractor strategy towards~$X_b$ for Player~$1$.
	We show that Player~$0$ does not have a strategy with cost at most~$b$ in~$\gamesup$ by showing that~$\tau_A$ has cost exceeding~$b$.
	We obtain~$\cost(\tau_A) > b$ in~$\gamesup$:
	By playing consistently with~$\tau_A$, Player~$1$ forces the play to eventually reach a vertex in~$X_b$, i.e., a vertex~$v$ with~$\rank(v) > b$.
	Thus,~$\cost(\tau_A) > b$, i.e.,~$\cost(\sigma) > b$ for all strategies~$\sigma$ of Player~$0$.
	
	Now assume that Player~$0$ does not have a winning strategy in~$\game'$.
	Towards a contradiction, assume that she has a strategy~$\sigma$ with cost at most~$b$ in~$\gamesup$.
	We first observe that no play consistent with~$\sigma$ visits any vertex from~$A$.
	Otherwise, playing consistently with his attractor strategy towards~$X_b$ from the first visit to~$A$, Player~$1$ would be able to construct a play consistent with~$\sigma$, but with cost greater than~$b$.
	Thus,~$\sigma$ is a strategy for Player~$0$ in~$\game'$ and we obtain that all plays consistent with~$\sigma$ in~$\arena$ are consistent with~$\sigma$ in~$\arena'$ and vice versa.
	Since~$\cost(\sigma) \leq b$, we obtain~$\suprankgame(\wincond, \rank)(\rho) < \infty$, i.e.,~$\rho \in \wincond$ for all plays~$\rho$ consistent with~$\sigma$.
	Thus,~$\sigma$ is a winning strategy for Player~$0$ in~$\game'$, a contradiction.
\qed
\end{proof}

Using this lemma, we are able to construct a decision procedure solving games from~$\supgames{\gamefam}$ using a decision procedure solving games from~$\gamefam$.

\begin{proof}[Proof of Theorem \ref{thm:vertex-ranked-games:direct:complexity}]
	Since~$\gamefam$ is proper,~$\supgames{\gamefam}$ is proper as well.
	Given the vertex-ranked~$\sup$-game~$\gamesup = (\arena, \suprankgame(\wincond, \rank))$, let~$X_b = \set{v \in V \mid \rank(v) > b}$ and let~$A = \att{1}{}{X_b}$.
	We define the decision procedure~$\decsup$ deciding the given problem  such that it returns false if~$v_\initmark \in A$.
	Otherwise,~$\decsup$ returns true if and only if Player~$0$ wins~$\game \setminus A$.
	Since~$\supgames{\gamefam}$ is proper and due to the assumption of the theorem, $\game \setminus A$ can be solved in time at most $t(\card{\game})$ and space at most $s(\card{\game})$.
	The procedure~$\decsup$ indeed decides the given decision problem due to Lemma~\ref{lem:vertex-ranked-games:direct:reduction}.
	
	Since we can compute and remove the Player-$1$-attractor~$A$ in linear time in~$\card{\arena}$~\cite{NerodeRemmelYakhnis96}, the decision procedure~$\decsup$ indeed requires time~$\bigo(\card{\arena}) + t(\card{\game})$ and space~$\bigo(\card{\arena}) + s(\card{\game})$.
\qed
\end{proof}

This theorem provides an upper bound on the complexity of solving vertex-ranked~$\sup$-games.
Intuitively, we prove Theorem~\ref{thm:vertex-ranked-games:direct:complexity} by showing that, for any vertex-ranked~$\sup$-game~$\gamesup$, a winning strategy for Player~$0$ in~$\game$ that never moves to the Player~$1$-attractor towards vertices of rank greater than~$b$ has cost at most~$b$.
Thus, an upper bound on the size of winning strategies for Player~$0$ for games from~$\gamefam$ provides an upper bound for strategies of finite cost in~$\supgames{\gamefam}$ as well.
Moreover, if the decision procedure deciding~$\gamefam$ constructs winning strategies for one or both players, we can adapt the decision procedure deciding~$\supgames{\gamefam}$ to construct strategies of cost at most (greater than)~$b$ for Player~$0$ (Player~$1$) as well.

\begin{corollary}
	\label{cor:vertex-ranked-games:sup:memory}
	Let~$\gamefam_{\sup}$ be a proper class of vertex-ranked~$\sup$-games and let~$\gamesup \in \gamefam_{\sup}$.
	If $\sigma$ is a finite-state winning strategy for Player~$i$ in~$\game$, then Player~$i$ has a finite-state winning strategy~$\sigma_{\sup}$ in~$\gamesup$ with~$\card{\sigma_{\sup}} \in \bigo(\card{\sigma})$.
	Furthermore, if~$\sigma$ is effectively constructible, then~$\sigma_{\sup}$ is effectively constructible.
\end{corollary}

Finally, the procedure constructed in the proof of Theorem \ref{thm:vertex-ranked-games:direct:complexity} enables us to solve the optimization problem for vertex-ranked $\sup$-games from~$\supgames{\gamefam}$:
Recall that if Player~$0$ wins~$\gamesup$ with respect to some~$b$, she wins it with respect to all~$b' \geq b$ as well.
Hence, using a binary search,~$\log(M)$ invocations of the decision procedure from the proof of Theorem~\ref{thm:vertex-ranked-games:direct:complexity} suffice to determine the minimal~$b$ such that Player~$0$ wins~$\gamesup$ with respect to~$b$, where~$M$ denotes the number of ranks assigned to vertices of~$\gamesup$ by its ranking function.
Hence, it is possible to determine the minimal such~$b$ in time $\bigo(\log(M)(n + t(\card{\game})))$ and space~$\bigo(M) + s(\card{\game})$.

\subsection{Solving Vertex-Ranked~$\limsup$-Games}
\label{sec:vertex-ranked-games:limsup}

We now turn our attention to solving vertex-ranked~$\limsup$-games.
Solving these games is again at least as hard as solving their corresponding qualitative games, due to the same reasoning as for vertex-ranked~$\sup$-games.
Thus, we again only provide upper bounds on the complexity of solving such games.
To this end, given some class~$\gamefam$ of games, we define the corresponding class of vertex-ranked~$\limsup$-games
\[
	\limsupgames{\gamefam} = \{ (\arena, \limsuprankgame(\wincond, \rank)) \mid (\arena, \wincond) \in \gamefam, \rank \text{ is vertex-ranking function for } \arena \} \enspace .
\]

\newcommand{\cobuchi}{\textsc{CoBüchi}}

We identify two criteria on classes of qualitative games~$\gamefam$, each of which is sufficient for quantitative games in~$\limsupgames{\gamefam}$ to be solvable with respect to some given~$b$.
More precisely, we provide decision procedures for~$\limsupgames{\gamefam}$ for the case that
\begin{itemize}
	\item games from~$\gamefam$ can be solved in conjunction with coBüchi-conditions, and for the case that
	\item the winner of a play~$\rho$ in a game from~$\gamefam$ depends only on an infinite suffix of~$\rho$.
\end{itemize}

The latter condition is commonly referred to as prefix-independence, which we formally define later in this section.

In order to show the former case, fix some class of games~$\gamefam$ and let $\gamelim = (\arena, \limsuprankgame(\wincond, \rank)) \in \limsupgames{\gamefam}$ be a vertex-ranked~$\limsup$-game with vertex set~$V$.
Furthermore, recall that a play in~$\wincond$ has cost at most~$b$ in~$\gamelim$ if it visits vertices of rank greater than~$b$ only finitely often.

In the qualitative case, the behavior of visiting a certain set of vertices only finitely often is formalized by the qualitative co-Büchi condition
\[
	\cobuchi(F) = \set{\rho \in V^\omega \mid \inf(\rho) \cap F = \emptyset} \enspace ,
\]
where~$\inf(\rho)$ denotes the set of vertices occurring infinitely often in~$\rho$.
Clearly, Player~$0$ has a strategy of cost at most~$b$ in~$\gamelim$ if and only if she wins~$(\arena, \wincond \cap \cobuchi(\set{v \in V \mid \rank(v) > b}))$.
This observation gives rise to the following remark.

\begin{remark}
Let~$\gamefam$ be a class of qualitative games such that the games in~$\set{(\arena, \wincond \cap \cobuchi(F)) \mid (\arena, \wincond) \in \gamefam, F \subseteq V, V \text{ is vertex set of } \arena}$ can be solved in time~$t(\card{\game}, \card{F})$ and space~$s(\card{\game}, \card{F})$, where~$t$ and~$s$ are monotonic functions.

Then, the following problem can be solved in time $t(\card{\gamelim}, n)$ and space $s(\card{\gamelim}, n)$:
\myquot{Given some game~$\gamelim \in \limsupgames{\gamefam}$ with~$n$ vertices as well as some bound~$b \in \nats$, does Player~$0$ win~$\gamelim$ with respect to~$b$?}
\end{remark}

In this case, we solve vertex-ranked~$\limsup$-games via a decision procedure for solving qualitative games as-is.
Such a procedure trivially exists if the winning conditions of games from~$\gamefam$ are closed under intersection with co-Büchi conditions.
Thus, we obtain solvability of a wide range of classes of vertex-ranked~$\limsup$-games, e.g., co-Büchi-, parity-, Muller-, Streett- and Rabin games.

We now turn our attention to the latter case described above:
We consider classes~$\limsupgames{\gamefam}$ where a play is only determined to be winning or losing in a game from~$\gamefam$ due to some infinite suffix.
Formally, we say that a qualitative winning condition~$\wincond \subseteq V^\omega$ is \emph{prefix-independent} if for all infinite plays~$\rho \in V^\omega$ and all play prefixes~$\pi \in V^*$, we have $\rho \in \wincond$ if and only if~$\pi\rho \in \wincond$.
A qualitative game is prefix-independent if its winning condition is prefix-independent.
A class of games is prefix-independent if every game in the class is prefix-independent.
This notion allows us to formalize the claim made in the second bullet point above.

\begin{theorem}
\label{thm:vertex-ranked-games:indirect:complexity}
Let~$\gamefam$ be a proper prefix-independent class of qualitative games where each~$\game\in\gamefam$ can be solved in time~$t(\card{\game})$ and space~$s(\card{\game})$, where~$t$ and~$s$ are monotonic functions.

Then, the following problem can be solved in time~$\bigo(n^3 + n^2\cdot t(\card{\gamelim}))$ and space $\bigo(n + s(\card{\gamelim}))$:
\myquot{Given some game~$\gamelim \in \limsupgames{\gamefam}$ with~$n$ vertices and some bound~$b \in \nats$, does Player~$0$ win~$\gamelim$ with respect to~$b$?}
\end{theorem}

Let~$\gamefam$ be a proper prefix-independent class of games and let~$\gamelim \in \limsupgames{\gamefam}$.
Moreover, let~$b \in \nats$.
Intuitively, in order to solve the~$b$-threshold problem for~$\gamelim$, we adapt the classic algorithm for solving prefix-independent qualitative games (cf., e.g., the work by Chatterjee, Henzinger, and Piterman~\cite{ChatterjeeHenzingerPiterman06}).
Thereby, we repeatedly compute the set of vertices from which Player~$0$ has a strategy of cost at most~$b$ in the corresponding vertex-ranked~$\sup$-game ~$\gamesup$ and remove their~$0$-attractor from the game similarly to the construction of a decision procedure for vertex-ranked $\sup$-games in the proof of Theorem~\ref{thm:vertex-ranked-games:direct:complexity}.
We claim that Player~$0$ has a strategy with cost at most~$b$ in~$\gamelim$ if and only if~$v_\initmark$ was removed during that above construction.

In order to prove Theorem~\ref{thm:vertex-ranked-games:indirect:complexity}, we first show that, if Player~$0$ does not win a~$\sup$-game from any vertex, then she also does not win the corresponding~$\limsup$-game from any vertex.
Recall that for a qualitative or quantitative game~$\game$ with vertex set~$V$ we write~$\game_v$ to denote the game~$\game$ with its initial vertex replaced by~$v \in V$.
All other components, i.e., the structure of the arena and the cost-function, remain unchanged.
We write~$\win^b_i(\game)$ to denote the set of all vertices~$v$ such that Player~$i$ has a strategy of cost at most~$b$, if~$i = 0$, or greater than~$b$, if~$i = 1$, in~$\game_v$.

\begin{lemma}
\label{lem:vertex-ranked-games:direct:completeness}
Let~$\gamelim = (\arena, \limsuprankgame(\wincond, \rank))$ be a vertex-ranked $\limsup$-game with vertex set~$V$ such that~$\wincond$ is prefix-independent and such that for each $v \in V$ the vertex-ranked $\sup$-game~$(\gamesup)_v$ is determined.
If~$\win^b_0(\gamesup) = \emptyset$, then~$\win^b_0(\gamelim) = \emptyset$.
\end{lemma}

\begin{proof}
Let~$V$ be the vertex set of~$\gamesup$ and~$\gamelim$.
Since~$\win^b_0(\gamesup) = \emptyset$ and since for all $v \in V$ the game $(\gamesup)_v$ is determined, we obtain~$\win^b_1(\gamesup) = V$.
For each~$v \in V$, let~$\tau'_v$ be a strategy for Player~$1$ in~$(\gamesup)_v$ with cost greater than~$b$.
We now define a single strategy~$\tau$ for Player~$1$ in~$\gamelim$ with cost greater than~$b$.
For each~$\pi = v_0\cdots v_j \in V^*$ we define~$\tau(\pi) = \tau'_{v_k}(v_k \cdots v_j)$, where~$k = \max \set{k' \mid \rank(v_{k'-1}) > b}$, with~$\max\emptyset = 0$.
We claim that~$\tau$ has cost greater than~$b$ in all~$(\gamelim)_v$.
Since the cost-function~$\cost$ is identical in all~$(\gamelim)_v$ this claim is formalized as~$\cost(\tau) > b$.

Let~$\rho = v_0v_1v_2\cdots$ be a play of~$\game_v$ consistent with~$\tau$.
If there are infinitely many positions~$j$ with~$\rank(v_j) > b$, then~$\cost(\rho) > b$.
Thus, assume the opposite and let~$j$ be the maximal position with~$\rank(v_j) > b$.
Then the suffix~$\rho' = v_{j+1}v_{j+2}v_{j+3}\cdots$ of~$\rho$ is consistent with~$\tau'_{v_{j+1}}$.
Since~$\rho'$ does not encounter any vertices of rank greater than~$b$, while~$\cost(\rho') > b$ due to~$\rho'$ being consistent with a strategy of cost greater than~$b$, we obtain~$\rho' \notin \wincond$.
This implies~$\rho \notin \wincond$ due to prefix-independence of~$\wincond$.
Hence,~$\cost(\rho) = \infty$, which, together with the statement above, implies~$\cost(\tau) > b$.
\qed
\end{proof}

We are now able to prove Theorem~\ref{thm:vertex-ranked-games:indirect:complexity} using Lemma~\ref{lem:vertex-ranked-games:direct:completeness} as a building block for showing the correctness of the approach outlined above.

\begin{proof}[Proof of Theorem~\ref{thm:vertex-ranked-games:indirect:complexity}]
	Given~$\gamelim = (\arena, \limsuprankgame(\wincond, \rank))$ with vertex set~$V$ of size~$n$, we define~$\game_0 = \gamesup$, as well as~$X_j = \win^b_0(\game_j)$, $A_j = \att{0}{}{X_j}$, which is computed in the arena of~$\game_j$, and $\game_{j+1} = \game_j \setminus A_j$ for all~$j \in \nats$.
	As we only remove vertices from the games~$\game_j$, we obtain~$\game_{j+1} \subarenaof \game_j$.
	Thus, the series of games stabilizes at~$j=n$ at the latest, i.e., $\game_j = \game_n$ for all~$j \geq n$.
	We define~$A = \bigcup_{j \leq n} A_j$ and~$\game' = \game_n$ and claim that Player~$0$ has a strategy with cost at most~$b$ in~$\game$ if and only if~$v_\initmark \in A$.
	We first argue that this suffices to show the desired result.
	
	First note that since $\gamefam$ is proper, $\supgames{\gamefam}$ is proper as well.
	Thus, Theorem~\ref{thm:vertex-ranked-games:direct:complexity} is applicable to $\supgames{\gamefam}$.
	Let~$\decsup$ be the decision procedure deciding whether or not Player~$0$ has a strategy with cost at most~$b$ in games from~$\supgames{\gamefam}$, as constructed in the proof of that theorem.
	The decision procedure~$\decsup$ can be easily modified to return~$W^b_0(\game_j)$ instead of a yes/no-answer by applying it to each~$(\game_j)_v$ individually.
	This, however, is only possible since we assume~$\gamefam$ to be proper, as the second condition of the definition of a proper family of games allows us to solve each $(\game_j)_v$.
	This modified procedure~$\decsup'$ runs in time at most~$\bigo(n^2 + n\cdot t(\card{\game}))$ and space $\bigo(n) + s(\card{\game})$, where~$t(\card{\game})$ and~$s(\card{\game})$ are the time and space required to solve~$\game$, respectively.
	
	For~$j \in \set{0,\dots,n}$, the decision procedure~$\declim$ first computes~$\game_j$ in linear time in~$n$ and reusing the space used for solving~$\game_{j-1}$.
	It then computes~$X_j$ requiring a single call to the modified~$\decsup$.
	It subsequently computes~$A_j$ in time~$\bigo(n)$ and space~$\bigo(n)$. 
	Finally, it returns false if and only if~$v_\initmark$ is in the arena of~$\game_n$.
	In total, we obtain a runtime of~$\declim$ of~$\bigo(n^3 + n^2\cdot t(\card{\game}))$.
	The only additional memory required by~$\declim$ is that for storing the sets~$X_j$ and~$A_j$, the size of which is bounded from above by~$n$.
	The games~$\game_j$ can be stored by reusing the memory occupied by~$\game$, due to~$\game_j \subarenaof \game_{j-1}$.
	Hence, the procedure~$\declim$ requires space~$\bigo(n) + s(\card{\game})$.

	It remains to show that Player~$0$ indeed has a strategy with cost at most~$b$ in~$\game$ if and only if~$v_\initmark \notin \arena_n$, i.e., if~$v_\initmark \in A$.
	To this end, first assume~$v_\initmark \in A$ and note that we have~$A_j \supseteq X_j$.
	However, for each two~$j \neq j'$, we have $A_j \cap A_{j'} = \emptyset$ and, in particular,~$X_j \cap X_{j'} = \emptyset$.
	Hence, for each~$v \in A$ there exists a unique~$j$ such that~$v \in A_j$.
	
	We define the strategy~$\sigma$ for Player~$0$ in~$\game$ inductively such that any play consistent with~$\sigma$ only descends through the~$X_j$.
	Formally, we construct~$\sigma$ such that it satisfies the following invariant:
	\begin{quote}
	Let~$\rho = v_0v_1v_2\cdots$ be a play consistent with~$\sigma$ and let~$k \in \nats$.
	If~$v_k \in A_j \setminus X_j$, then~$v_{k+1} \in \bigcup_{j' \leq j} A_{j'} \cup X_{j'}$.
	Moreover, if~$v_k \in (A_j \setminus X_j) \cap V_0$, then the move to~$v_{k+1}$ is the move prescribed by the attractor strategy of Player~$0$ towards~$X_j$.
	If~$v_k \in X_j$, then~$v_{k+1} \in X_j \cup \bigcup_{j' < j} A_{j'} \cup X_{j'}$.
	\end{quote}
	Clearly, this invariant holds true for~$\pi = v_\initmark$.
	Thus, let~$\pi = v_0\cdots v_k$ be a play prefix consistent with~$\sigma$.
	If~$v_k \in V_1$, let~$v^*$ be an arbitrary successor of~$v_k$ in~$\game$ and assume towards a contradiction that~$\pi v^*$ violates the invariant.
	If~$v_k \in A_j \setminus X_j$, then in~$\game_j$ there exists an edge from~$v_k$ leading to some vertex~$v^* \notin A_j$, a contradiction to the definition of the attractor.
	If, however,~$v_k \in X_j$ and~$v^* \notin X_j \cup \bigcup_{j' < j} A_{j'} \cup X_{j'}$, then Player~$1$ has a strategy~$\tau$ in~$(\game_j)_{v^*}$ with cost greater than~$b$.
	Thus, a play that begins in~$v_k$, moves to~$v^*$ and is consistent with~$\tau$ afterwards has cost greater~$b$, i.e., Player~$0$ does not have a strategy with cost at most~$b$ in~$(\game_j)_{v_k}$, a contradiction to~$v_k \in X_j = \win^b_0(\game_j)$.
	Hence,~$\pi v^*$ satisfies the invariant for each successor~$v^*$ of~$v_k \in V_1$.
	
	Now assume~$v_k \in V_0$ and first let~$v_0 \in A_j \cup X_j$ for some~$j \in \nats$.
	Let~$\sigma^A_j$ be an attractor strategy for Player~$0$ towards~$X_j$.
	If~$v_k \in A_j \setminus X_j$, we define~$\sigma(\pi) = \sigma^A_j(v_k)$, which satisfies the invariant due to the definition of the attractor strategy.
	If, however,~$v_k \in X_j$, let~$k'$ be minimal such that~$v_{k''} \in X_j$ for all~$k''$ with~$k' \leq k'' \leq k$.
	Moreover, let~$\sigma^v_j$ be a strategy for Player~$0$ such that every play consistent with~$\sigma^v_j$ in~$\game_j$ with initial vertex~$v$ has cost at most~$b$.
	Such a strategy exists due to~$X_j = \win^b_0(\game_j)$.
	We define~$\sigma(\pi) = \sigma^{v_{k'}}_j(v_{k'}\cdots v_k)$, which satisfies the invariant to similar reasoning as above.
	
	In order to show~$\cost(\sigma) \leq b$, let~$\rho = v_0v_1v_2\cdots$ be a play consistent with~$\sigma$.
	Due to the invariant of~$\sigma$ and since~$v_0 \in A$, the play~$\rho$ descends through the~$A_j$ and the~$X_j$, i.e., once it encounters some~$X_j$, it never moves to any~$A_{j'} \setminus X_{j'}$ with~$j' \geq j$ nor to any~$X_{j'}$ with~$j' > j$.
	Also,~$\rho$ stabilizes in some~$X_j$, i.e., there exists a~$k \in \nats$ such that~$v_{k'} \in X_j$ for all~$k' \geq k$, as~$\sigma$ prescribes moves according to the attractor strategy towards~$X_j$ when in~$A_j \setminus X_j$.
	Moreover, due to the definition of~$\sigma$, the suffix~$\rho' = v_{k}v_{k+1}v_{k+2}\cdots$ is consistent with~$\sigma^{v_k}_j$, i.e., we obtain~$\rho' \in \wincond$ and that the maximal vertex-rank encountered in~$\rho$ is at most~$b$.
	As~$\wincond$ is prefix-independent, we obtain~$\rho \in \wincond$ as well as~$\limsup_{k \rightarrow \infty}\rank(v_k) \leq b$.
	Hence,~$\limsuprankgame(\wincond, \rank)(\rho) \leq b$, which concludes this direction of the proof.
	
	Now assume~$v_\initmark \notin A$ and consider~$\game'$ with vertex set~$V\setminus A$.
	Since the construction of the~$\game_j$ stabilized, we have~$A_j = X_j = W^b_0(\game') = \emptyset$, i.e., Player~$1$ has a strategy with cost greater than~$b$ from any starting vertex in~$\game'$.
	Due to Lemma~\ref{lem:vertex-ranked-games:direct:completeness}, this implies that he has such a strategy from every vertex in~$\gamesup \setminus A$, call it~$\tau$.
	Note that there exists no Player-$0$-vertex in~$V \setminus A$ that has an outgoing edge leading into~$A$, as this would contradict the definition of the Player-$0$-attractors~$A_j$.
	Hence,~$\tau$ is a strategy for Player~$1$ in~$\game$ as well and we retain~$\cost(\tau) > b$.
\qed
\end{proof}

Intuitively, we prove Theorem~\ref{thm:vertex-ranked-games:indirect:complexity} by constructing a strategy~$\sigma$ for Player~$0$ by \myquot{stitching together} the attractor-strategies towards her winning regions in the decreasing vertex-ranked~$\sup$-games and the winning strategies for her in the respective vertex-ranked~$\sup$-games.
As each play consistent with that strategy descends down the hierarchy of~$\sup$-games thus constructed, we can reuse the memory states of the winning strategies in these games when implementing~$\sigma$.
Thus, a monotonic upper bound on the size of strategies with cost at most~$b$ in~$\gamesup$ is an upper bound on the size of such strategies in~$\gamelim$ as well.

\begin{corollary}
\label{cor:vertex-ranked-games:lim:memory}
Let~$\gamefam$ be a proper prefix-independent class of qualitative games such that, if Player~$0$ wins~$\game$, then she has a finite-state winning strategy of size at most~$m(\card{\game})$, where~$m$ is a monotonic function.

If Player~$0$ wins~$\gamelim \in \limsupgames{\gamefam}$, then she has a finite-state winning strategy~$\sigma_{\lim}$ with $\card{\sigma_{\lim}} \in \bigo(m(\card{\gamelim}))$ in~$\gamelim$.
Furthermore, if winning strategies for Player~$0$ in the games in~$\gamefam$ are effectively constructible, then~$\sigma_{\lim}$ is effectively constructible.
\end{corollary}

Moreover, in order to find the optimal~$b$ such that Player~$0$ wins~$\gamelim$ with respect to~$b$, we can again employ a binary search analogously to the case of vertex-ranked~$\sup$-games.
Thus, we can determine the optimal such~$b$ in time~$\bigo(\log(M)(n^3 + n^2\cdot t(\card{\gamelim})))$ and space $\bigo(n + s(\card{\gamelim}))$, where~$M$ again denotes the number of ranks assigned to vertices in a given vertex-ranked $\limsup$-game.

Having thus defined both quantitative reductions and a canonical target for such reductions, we now give examples of how to solve quantitative games using this tools.

%% file: applications.tex
\section{Applications}
\label{sec:application}

In this section, we give examples of how to use quantitative reductions and vertex-ranked games to solve quantitative games.
First, in Section~\ref{sec:applications:request-response}, we formally introduce a quantitative variant of request-response games, which we call request-response games with costs, and show how to solve such games using quantitative reductions and vertex-ranked sup-request-response games.
Second, in Section~\ref{sec:applications:muller}, we recall the definition of quantitative Muller games due to McNaughton~\cite{McNaughton00} and show how to reduce such games to vertex-ranked safety games via quantitative reductions.
Finally, in Section~\ref{sec:application:resilience}, we show that vertex-ranked games are useful in their own right, by showing how to use them to synthesize controllers that are resilient against disturbances.

\input{applications-reqres}

\input{applications-muller}

\input{applications-resilience}

%% file: applications-reqres.tex
\subsection{Reducing Request-Response Games with Costs to Vertex-Ranked Request-Response Games}
\label{sec:applications:request-response}

Recall that a play satisfies the qualitative request-response condition if every request that is opened is eventually answered.
We extend this condition to a quantitative one by equipping the edges of the arena with costs and measuring the maximal cost incurred between opening and answering a request. 

Fix some arena~$\arena$ with vertex set~$V$ and set~$E$ of edges.
Formally, the qualitative request-response condition~$\reqres(\Gamma)$ consists of a family of so-called request-response pairs~$\Gamma = (Q_c, P_c)_{c \in [d]}$, where ~$d \in \nats$, $d > 0$, and where $Q_c,P_c \subseteq V$ for all~$c \in [d]$.
Player~$0$ wins a play according to this condition if each visit to some vertex from~$Q_c$ is answered by some later visit to a vertex from~$P_c$, i.e., we define
\begin{multline*}
	\reqres((Q_c, P_c)_{c \in [d]}) = \\ \{ v_0v_1v_2\cdots \in V^\omega \mid \forall c \in [d] \forall j \in \nats.\, v_j \in Q_c \text{ implies } \exists j' \geq j.\, v_{j'} \in P_c\} \enspace.
\end{multline*}
We say that a visit to a vertex from~$Q_c$ \emph{opens a request for condition}~$c$ and that the first visit to a vertex from~$P_c$ afterwards \emph{answers the request for that condition}.

\begin{proposition}[\cite{WallmeierHuettenThomas03}]\label{thm:request-response:prior-work}
Request-response games with~$n$ vertices and~$d$ request-response pairs can be solved in time $\bigo(n^2 d^2 2^d)$.

Furthermore, let~$\game$ be a request-response game with~$d$ request-response pairs.
If Player~$0$ has a winning strategy in~$\game$, then she has a finite-state winning strategy of size at most~$d2^d$.
\end{proposition}

We extend this qualitative winning condition to a quantitative one using families of cost functions~$\cost = (\cost_c)_{c \in [d]}$, where~$\cost_c\colon E \rightarrow \nats$ for each~$c \in [d]$ and lift the cost functions~$\cost_c$ to play infixes~$\pi$ in~$\arena$ by adding up the costs along~$\pi$.
The cost-of-response for a request for condition~$c$ at position~$j$ is then defined as
\begin{multline*}
	\reqresdist_c(v_0v_1v_2\cdots, j) = \\
	\begin{cases}
 		\min \set{ \cost_c (v_j \cdots v_{j'}) \mid j' \ge j \text{ and } v_{j'} \in P_c } &\text{if } v_j \in Q_c \enspace , \\
 		0 &\text{otherwise} \enspace ,
 	\end{cases}
\end{multline*}
with~$\min \emptyset = \infty$, which naturally extends to the (total) \emph{cost-of-response}
\[
	\reqresdist(\rho, j) = \max\nolimits_{c \in [d]} \reqresdist_c(\rho, j) \enspace.
\]
Finally, we define the \emph{request-response condition} with costs as 
\[
	\costreqres(\Gamma, \cost)(\rho) = \sup\nolimits_{j\rightarrow \infty} \reqresdist(\rho , j) \enspace,
\]
i.e., this condition measures the maximal cost incurred by any request in~$\rho$.

We call a game~$\game = (\arena, \costreqres(\Gamma, \cost))$ a request-response game with costs.
We denote the largest cost assigned to any edge by~$W$.
As we assume the functions~$\cost_c$ to be given in binary encoding, the largest cost~$W$ assigned to an edge may be exponential in the size of~$\game$.

If all~$\cost_c$ assign zero to every edge, then the request-response condition with costs coincides with the qualitative request-response condition.
In general, however, the request-response condition with costs is a strengthening of the classic request-response condition:
If some play~$\rho$ has finite cost according to the condition with costs, then it is winning for Player~$0$ according to the qualitative condition, but not vice versa.

\begin{remark}
\label{rem:request-response:cost-to-classical}
Let~$\game = (\arena, \costreqres(\Gamma, \cost))$ be a request-response game with costs.
If a strategy~$\sigma$ for Player~$0$ in~$\game$ has finite cost, then~$\sigma$ is a winning strategy for Player~$0$ in the qualitative game~$(\arena, \reqres(\Gamma))$.
\end{remark}

This remark together with a detour via qualitative request-response games yield a cap for request-response games with costs.

\begin{lemma}
\label{lem:request-response:cap}
Let~$\game$ be a request-response game with costs with~$n$ vertices,~$d$ request-response pairs, and largest cost of an edge~$W$.
If Player~$0$ has a strategy with finite cost in~$\game$, then she also has a strategy with cost at most $d 2^d n W$.
\end{lemma}

\begin{proof}
Let~$\game = (\arena, \costreqres(\Gamma, \cost))$ and let~$\game' = (\arena, \reqres(\Gamma))$ be a qualitative request-response game obtained by disregarding the cost functions of~$\game$.
Moreover, let~$\sigma$ be a strategy with finite cost for~$\game$.
Due to Remark~\ref{rem:request-response:cost-to-classical}, the strategy~$\sigma$ is winning for Player~$0$ in~$\game'$ as well, hence Player~$0$ wins~$\game'$.
Thus, due to Proposition~\ref{thm:request-response:prior-work}, she has a winning strategy~$\sigma'$ of size at most~$d2^d$ in~$\game'$.
Let~$\sigma'$ be implemented by the memory structure~$\mem$ and let~$b = d2^dnW$.
We show~$\costreqres(\Gamma, \cost)(\sigma') \leq b$.

Let~$\rho = v_0v_1v_2\cdots$ be a play consistent with~$\sigma'$ and assume towards a contradiction $\costreqres(\Gamma, \cost)(\rho) > b$.
Then there exist~$c \in [d]$ and~$j \in \nats$ such that~$\reqresdist_c(\rho, j) > b$.
As each edge has cost at most~$W$, the request for condition~$c$ opened at position~$j$ is not answered for at least~$d2^dn$ steps, i.e., we obtain $v_{j'} \notin P_c$ for all~$j'$ with~$j \leq j' \leq j+d2^dn$.
Let~$\ext(\rho) = (v_0,m_0)(v_1,m_1)(v_2,m_2)\cdots$.
Since~$\card{\mem} \leq d2^d$, there exists a vertex repetition on the play infix~$(v_j,m_j)\cdots(v_{j+d2^dn},m_{j+d2^dn})$ of~$\ext(\rho)$, say at positions~$k$ and~$k'$ with $j \leq k < k' \leq j + d2^dn$.
Thus, the play~$\rho' = v_0\cdots v_k(v_{k+1}\cdots v_{k'})^\omega$ is consistent with~$\sigma'$.

In~$\rho'$, however, a request for condition~$c$ is opened at position~$j \leq k$.
Since we have~$j \leq k' \leq j+d2^dn$, this request is not answered in the play infix $v_j\cdots v_k\cdots v_{k'}$, i.e., it is never answered.
Hence,~$\rho' \notin \reqres(\Gamma)$, which contradicts~$\sigma'$ being a winning strategy for Player~$0$ in~$\game'$.
\qed
\end{proof}

Having obtained a cap for request-response games with costs, we can now turn to the main result of this section:
Request-response games with costs are reducible to vertex-ranked $\sup$-request-response games.
In order to show this, we use a memory structure that keeps track of the costs incurred by the requests open at each point in the play~\cite{WeinertZimmermann17}.

\begin{lemma}
\label{lem:request-response:reduction}
	Let~$\game$ be a request-response game with costs with~$n$ vertices,~$d$ request-response pairs, and highest cost of an edge~$W$.
	Then~$\game \reducesto^{b+1}_\mem \game'$ for $b = d2^d nW$, some memory structure~$\mem$ of size~$\bigo(2nb^d)$, and a vertex-ranked~$\sup$-request-response game~$\game'$ with~$d$ request-response pairs.
\end{lemma}

\begin{proof}
	Let~$\game = (\arena, \costreqres(\Gamma, \cost))$ with initial vertex~$v_\initmark$.
	Recall that~$b$ is a cap of~$\game$ due to Lemma~\ref{lem:request-response:cap}.
	We first define the memory structure~$\mem$.
	Intuitively, we use it to keep track of the currently open requests and the costs they have incurred up to the cap~$b$.
	Once the cost of a single request incurs a cost greater than~$b$, the memory structure raises a Boolean flag, which indicates that Player~$1$ can unbound the cost of that request.
	
	Let~$r\colon [d] \rightarrow \set{\bot} \cup [b+1] = \set{\bot, 0, \dots, b}$ be a function mapping conditions~$c$ to the cost~$r(c) \in [b+1]$ they have incurred so far, or to~$r(c) = \bot$ if no request for that condition is pending.
	We call such a function a \emph{request-function} and denote the set of all request functions by~$R$.
	We define the initial request function~$r_\initmark$ such that~$r_\initmark(c) = 0$ if~$v_\initmark \in Q_c$ and~$r_\initmark(c) = \bot$ otherwise.
	In order to be able to access the current vertex during the update of the memory structure, we store it in the memory structure as well.
	By accessing the current vertex together with the vertex that we move to, we are thus able to obtain the cost of the traversed edge.
	Finally, we store a flag that indicates whether or not the bound~$b$ has been exceeded.
	Hence, we define the set of memory states $M = V \times R \times \set{0,1}$ with the initial memory state $m_\initmark = (v_\initmark, r_\initmark, 0)$.
	
	We define the update function~$\update((v, r, f), v') = (v', r', f')$ by performing the following steps in order:
	\begin{itemize}
		\item For each $c \in [d]$, if $v \in P_c$, set $r'(c) = \bot$. Otherwise, set $r'(c) = r(c)$.
		\item For each~$c \in [d]$, if~$r'(c) \neq \bot$, set $r'(c) = r'(c) + \cost_c((v,v'))$.
		\item Now, if there exists a condition $c$ such that $r'(c) > b$, then set~$r'(c) = \bot$ for all~$c$ and set $f'= 1$.
			Otherwise, set~$f' = f$.
		\item For each~$c \in [d]$, if $v' \in Q_c$, set~$r'(c)$ to~$\max\set{r'(c),0}$ where~$\max\set{\bot, 0} = 0$.
	\end{itemize}
	Note that we replicate the current vertex~$v$ in the memory state in order to be able to access it in the update of the memory state during the move to~$v'$, thereby attaining access to the traversed edge~$(v, v')$.
	
	We obtain~$\mem = (M, m_\initmark, \update)$.
	Clearly, we have~$\card{\mem} \in \bigo(2nb^d)$, i.e.,~$\mem$ is of exponential size in~$d$, but only of polynomial size in~$n$ and~$W$.
	
	Using this definition, we obtain that if~$\costreqres(\Gamma, \cost)(\rho) \leq b$, then the extended play~$\ext(\rho)$ remains in vertices of the form $(v,v,r,0)$.
	Dually, if $\costreqres(\Gamma, \cost)(\rho) > b$, then~$\ext(\rho)$ eventually moves to vertices of the form $(v,v,r,1)$ and remains there ad infinitum.
	
	\newcommand{\costcrr}{\cost_\game}
	\newcommand{\costvrr}{\cost_{\game'}}
	
	Let~$\Gamma = (Q_c, P_c)_{c \in [d]}$.
	It remains to define the vertex-ranking function~$\rank\colon V \times M \rightarrow \nats$, as well as the family of request-response pairs~$\Gamma'$ for~$\game'$.
	We define the former as 
	\[
		\rank(v,v,r,f) = \begin{cases}
			\max\set{0, \max_{c \in [d]}r(c)} &\text{if } f = 0 \text{ and} \\
			b+1 &\text{otherwise}
		\end{cases}
	\]
	and the latter as
	\begin{multline*}
		\Gamma' = (Q'_c, P'_c)_{c \in [d]}\text{, where }Q'_c = Q_c \times Q_c \times R \times \set{0,1}\text{ and }\\
		P'_c = P_c \times P_c \times R \times \set{0,1}\text{ for all }c \in [d] \enspace .
	\end{multline*}
	Note that $\rho \in \reqres(\Gamma)$ if and only if~$\ext(\rho) \in \reqres(\Gamma')$.
	
	We first argue that the ranking functions included in the memory state indeed contain information about the cost incurred so far by open requests as long as the cost of the play does not exceed~$b$.
	In order to do so, let~$\rho = v_0v_1v_2\cdots$ be some play in~$\game$ with $\costreqres(\Gamma, \cost)(\rho) \leq b$ and let $\ext(\rho) = (v_0,v_0,r_0,f_0)(v_1,v_1,r_1,f_1)(v_2,v_2,r_2,f_2)\cdots$ be its extension.
	Intuitively, since the cost of~$\rho$ does not exceed the bound~$b$, the flags~$f_j$ are not raised and the request functions track the cost of all requests precisely.
	
	More formally if, for some~$c \in [d]$ and some~$j \in \nats$ with~$v_j \in Q_c$, we have $\reqresdist_c(\rho, j) = b'$, then $r_{j'}(c) = b'$, where~$j' \geq j$ is the earliest position at which the request for~$c$ opened at position~$j$ is answered.
	Dually, if $r_{j'}(c) = b'$ for some~$j' \in \nats$, some~$b' \in \nats$, and some~$c \in [d]$, then $\reqresdist_c(\rho, j) \geq b'$, where~$j \leq j'$ is the earliest position at which the request for condition~$c$ is opened without being answered prior to position~$j'$.
	In particular, if $r_{j'}(c) = b'$ with~$j'$,~$b'$, and~$c$ as above and we additionally have $v_{j'} \in P_c$, then $\reqresdist_c(\rho, j) = b'$.
	
	We define $\game' = (\arena \times \mem, \suprankgame(\reqres(\Gamma'), \rank))$.
	Moreover, since~$\Gamma'$ is the extension of~$\Gamma$ to the vertices of $\arena \times \mem$, the game~$\game'$ contains~$d$ many request-response pairs.
	
	It remains to show $\game \reducesto^{b+1}_\mem \game'$.
	Recall that, since we do not name a $(b+1)$-correction function explicitly, we implicitly use the $(b+1)$-correction-function~$\capfunc_{b+1}$.
	Clearly, the first and second condition of the definition of a quantitative reduction hold true, i.e., the arena of~$\game'$ is~$\arena \times \mem$ and~$\capfunc_{b+1}$ is a~$(b+1)$-correction function.
	It remains to show the two latter conditions.
	To this end, let $\rho = v_0v_1v_2\cdots \in V^\omega$ be a play in~$\game$ and let $\ext(\rho) = (v_0,v_0,r_0,f_0)(v_1,v_1,r_1,f_1)(v_2,v_2,r_2,f_2)\cdots$ be its unique extended play in~$\game'$.
	We use the shorthands $\costcrr = \costreqres(\Gamma, \cost)$ as well as $\costvrr = \suprankgame(\reqres(\Gamma'), \rank)$.
	
	We first show $\costcrr(\rho) = \costvrr(\ext(\rho))$ for all~$\rho$ with $\costcrr(\rho) < b+1$.
	Let $\costcrr(\rho) = b' < b+1$ and note that this implies $\rho \in \reqres(\Gamma)$ and $\ext(\rho) \in \reqres(\Gamma')$ as well as $f_j = 0$ for all $j \in \nats$.
	As argued above, we obtain $\rank(v_j,v_j,r_j,f_j) \leq b'$ for all~$j$, which implies~$\costvrr(\ext(\rho)) \leq b'$.
	Moreover, let~$c \in [d]$ and~$j \in \nats$ with~$v_j \in Q_c$ such that $\reqresdist_c(\rho, j) = b'$.
	Since~$b' < \infty$, such~$c$ and~$j$ exist.
	The play~$\ext(\rho)$ visits a vertex of rank~$b'$ at the position at which the request for condition~$c$ opened at position~$j$ is answered for the first time.
	Thus, $\costvrr(\ext(\rho)) \geq b'$, which concludes this part of the proof.
	
	It remains to show that $\costvrr(\ext(\rho)) \geq \capfunc_{b+1}(b+1) = b+1$ holds true for all~$\rho$ with~$\costcrr(\rho) \geq b+1$.
	To this end, let $\costcrr(\rho) = b' \geq b+1$.
	As argued above, the extended play~$\ext(\rho)$ eventually moves to vertices of the form $(v, v, r, 1)$ and remains there.
	Hence,~$\costvrr(\ext(\rho)) = b+1$ if~$\rho \in \reqres(\Gamma)$, i.e., if $\ext(\rho) \in \reqres(\Gamma')$.
	If, however, $\rho \notin \reqres(\Gamma)$, then $\ext(\rho) \notin \reqres(\Gamma')$ and hence, $\costvrr(\rho) = \infty > b+1$.
	\qed
\end{proof}

Thus, in order to solve a request-response game with costs with respect to some~$b$, it suffices to solve a vertex-ranked $\sup$-request-response game with respect to~$b$.
This, in turn, can be done by reducing the problem to that of solving a request-response game as shown in Theorem~\ref{thm:vertex-ranked-games:direct:complexity}.
Using this reduction together with the framework of quality-preserving reductions, we are able to provide an upper bound on the complexity of solving request-response games  with costs with respect to some bound~$b$.

\begin{theorem}
	\label{thm:request-response:exptime-membership}
	The following decision problem is in \exptime:
	\myquot{Given some request-response game with costs~$\game$ and some bound~$b \in \nats$, does Player~$0$ have a strategy~$\sigma$ with~$\cost(\sigma) \leq b$ in~$\game$?}
\end{theorem}

\begin{proof}
	Let~$\game$ contain~$n$ vertices,~$d$ request-response pairs, and let~$W$ be the largest cost assigned to any edge.
	We first construct the vertex-ranked~$\sup$-request-response game~$\game'$ from~$\game$ as shown in Lemma~\ref{lem:request-response:reduction}.
	Recall that~$\game'$ contains~$\bigo(n(d2^dnW)^d)$ vertices and~$d$ request-response pairs.
	Due to the instantiation of Theorem~\ref{thm:vertex-ranked-games:direct:complexity} with the decision procedure for qualitative request-response games from Proposition~\ref{thm:request-response:prior-work}, the game~$\game'$ can be solved with respect to~$b$ in time
	\begin{multline*}
		\bigo( n(d2^dnW)^d +( n(d2^dnW)^d))^2 d^2 2^d ) = \\ \bigo( n(d2^dnW)^d +n^2(d2^dnW)^{2d} d^2 2^d ) = \bigo( n^2(d2^dnW)^{2d} d^2 2^d ) \enspace .
	\end{multline*}
	Due to~$W^{2d} \in \bigo((2^{\card{\game}})^{\card{\game}}) = \bigo(2^{\card{\game}^2})$, this is exponential in the description length of~$\game$.
	\qed
\end{proof}

Moreover, solving request-response games is known to be \exptime-hard~\cite{ChatterjeeHenzingerHorn11}.
Thus, solving quantitative request-response games with costs via quantitative reductions is asymptotically optimal.

Furthermore, by leveraging our results on the sizes of memory structures in vertex-ranked~$\sup$-games we obtain an upper bound on the size of strategies with a given cost in request-response games with costs.

\begin{lemma}
Let~$\game$ be a request-response game with costs with~$n$ vertices,~$d$ request-response pairs, and largest cost of an edge~$W$.
If Player~$0$ has a strategy in~$\game$ with finite cost, then she also has a strategy in~$\game$ with finite cost of size at most~$\bigo(n b^d d 2^d)$, where~$b = d 2^d n W$.
\end{lemma}

\begin{proof}
Let~$\sigma$ be a strategy for Player~$0$ in~$\game$ with finite cost.
Due to Lemma~\ref{lem:request-response:cap}, Player~$0$ has a strategy~$\sigma'$ in~$\game$ with cost at most~$d2^dnW$.
Let~$\game'_{\sup}$ be the vertex-ranked $\sup$-request-response game constructed in the proof of Lemma~\ref{lem:request-response:reduction} and recall that~$\game'_{\sup}$ has~$d$ request-response pairs as well.

Due to Theorem~\ref{thm:capped-reduction}, and since~$b$ is a cap of~$\game$ due to Lemma~\ref{lem:request-response:cap}, Player~$0$ has a strategy~$\sigma'_{\sup}$ with cost at most $d2^dnW$ in~$\game'_{\sup}$.
Let~$\game'$ be the qualitative request-response game corresponding to~$\game'_{\sup}$, i.e., the game played in the same arena as~$\game'_{\sup}$ in which a play is winning for Player~$0$ if and only if it has finite cost in~$\game'_{\sup}$.

Clearly, the strategy~$\sigma'_{\sup}$ is winning for Player~$0$ in~$\game'$ as well.
By applying Proposition~\ref{thm:request-response:prior-work} we obtain that Player~$0$ has a winning strategy of size at most~$d2^d$ in~$\game'$.
By furthermore applying Corollary~\ref{cor:vertex-ranked-games:sup:memory} and Theorem~\ref{thm:reductions:strategy-lift}, we obtain that Player~$0$ has a strategy of finite cost of size $\bigo(2nb^d d2^d) = \bigo(n b^d d 2^d)$.
\qed
\end{proof}

Finally, the optimization problem of finding the minimal~$b'$ such that Player~$0$ wins a request-response game with costs~$\game$ with respect to~$b'$ can be solved in exponential time as well.
Recall that if Player~$0$ wins~$\game$ with respect to some~$b'$, then she also wins it with respect to all~$b'' \geq b'$.
Since we can assume~$b' \leq b = d 2^d n W$, we can perform a binary search for~$b'$ on the interval $\set{0,\dots,b}$.
Hence, the optimal~$b'$ can be found in time~$\bigo(\log(b) ( n^2 b^{2d} d^2 2^d ) )$.

%% file: applications-muller.tex
\subsection{Reducing Quantitative Muller Games to Vertex-Ranked Safety Games}
\label{sec:applications:muller}

Having shown how our framework can be used to find optimal strategies in request-response games with costs in a structured and modular way, we now show how it can be used to greatly simplify existing methods for finding such strategies.
To this end, we show how to reduce quantitative Muller games to vertex-ranked safety games, i.e., vertex-ranked games in which it is the aim of Player~$0$ to avoid a certain set of undesirable vertices.
In order to do so, we leverage techniques introduced by Neider, Rabinovich, and Zimmermann~\cite{NRZ14}.

Let~$\arena$ be some arena with vertex set~$V$ and recall that the qualitative Muller condition is defined via a partition of~$2^V$ into~$(\mullerfam_0, \mullerfam_1)$ as
\[
	\muller(\mullerfam_0, \mullerfam_1) = \set{ \rho \in V^\omega \mid \inf(\rho) \in \mullerfam_0 } \enspace ,
\]
where~$\inf(\rho)$ denotes the set of vertices that are visited infinitely often by~$\rho$.
Thus, Player~$i$ wins~$\rho$ if and only if~$\inf(\rho) \in \mullerfam_i$.

McNaughton introduced a quantitative characterization of the Muller condition by assigning a score to each prefix of a play and each subset of the set of vertices~\cite{McNaughton00}.
In order to characterize the set of vertices visited infinitely often during a play, the score of a subset~$F$ measures how often~$F$ has been visited completely without leaving it.
For a play~$\rho$, the limit inferior of the score of~$\inf(\rho)$ tends towards infinity, while the limit inferior of the score for all other sets is zero~\cite{McNaughton00}.

Formally, for any set~$F \subseteq V$ with~$F \neq \emptyset$, the score~$\score_F(\pi)$ is defined inductively using an accumulator that stores the vertices of~$F$ that have already been visited, as follows:
\begin{align*}
	(\acc_F(\epsilon),\score_F(\epsilon)) &= (\emptyset,0) \\
	(\acc_F(\pi v),\score_F(\pi v)) &= \begin{cases}
		(\emptyset, 0) &\text{if } v \notin F \\
		(\emptyset, \score_F(\pi) + 1) & \text{if } v \in F \text{ and } \acc_F(\pi) = F \setminus \set{v} \\
		(\acc_F(\pi) \cup \set{v}, \score_F(\pi)) & \text{otherwise}
	\end{cases}
\end{align*}

We generalize the score-function to families
$\mullerfam$ of subsets of vertices, i.e.,~$\mullerfam \subseteq 2^V$, by defining~$\score_\mullerfam(\pi) = \max_{F \in \mullerfam}(\score_F(\pi))$ and to infinite plays by defining~$\score_\mullerfam(v_0v_1v_2\cdots) = \sup_{j \rightarrow \infty}\score_\mullerfam(v_0\cdots v_j)$.
This definition inspires the quantitative Muller condition, which is defined as 
\[
	\quantmuller(\mullerfam_0,\mullerfam_1)(\rho) = \score_{\mullerfam_1}(\rho) \enspace .
\]
We obtain a cap for such games via leveraging a result by Fearnley and Zimmermann~\cite{FearnleyZimmermann12}.

\begin{lemma}
\label{lem:muller:cap}
Let~$\game = (\arena, \cost)$ be a quantitative Muller game.
If Player~$0$ has a strategy~$\sigma$ with finite cost in~$\game$, then she has a strategy~$\sigma'	$ with~$\cost(\sigma') \leq 2$.	
\end{lemma}

\begin{proof}
	Let~$\game = (\arena, \quantmuller(\mullerfam_0, \mullerfam_1))$.
	Since~$\cost(\sigma) < \infty$, for every play~$\rho$ consistent with~$\sigma$ and every prefix~$\pi$ of~$\rho$, we have that there exists an upper bound on~$\score_F(\pi)$ for all~$F \in \mullerfam_1$.
	Moreover, as the score of~$\inf(\rho)$ tends towards~$\infty$, this implies~$\inf(\rho) \in \mullerfam_0$, i.e.,~$\sigma$ is a winning strategy for the qualitative Muller game~$\game' = (\arena, \muller(\mullerfam_0, \mullerfam_1))$.
	
	It is known that, since Player~$0$ wins~$\game'$, she has a strategy~$\sigma'$ with $\score_{\mullerfam_1}(\pi) \leq 2$ for all prefixes~$\pi$ of all plays consistent with~$\sigma'$~\cite{FearnleyZimmermann12}.
	Thus, we directly obtain $\quantmuller(\mullerfam_0,\mullerfam_1)(\sigma') \leq 2$.
	\qed
\end{proof}

We now show how to reduce quantitative Muller games to vertex-ranked $\sup$-safety games based on previous work by Neider, Rabinovich, and Zimmermann~\cite{NRZ14}.
Recall that a safety game is a very simple qualitative game, in which it is Player~$0$'s goal to avoid a certain set of undesirable vertices.
Formally, the safety condition is defined via a set~$U \subseteq V$ as
\[ \safety(U) = \set{v_0v_1v_2\cdots \in V^\omega \mid \forall j \in \nats v_j \notin U} \enspace . \]

In order to construct the safety game, we define an equivalence relation over play prefixes, such that two play prefixes are equivalent if they have the same accumulator and the same score with respect to all~$F \in \mullerfam_1$.
The constructed safety game uses as vertices representatives of the equivalence classes of all play prefixes that have a cost of at most~$2$ for all~$F \in \mullerfam_1$.
Moreover, it mimics play prefixes~$\pi$ of cost at most~$2$ in the Muller game by moving to some vertex~$\pi'$ such that the score and the accumulator are equal in~$\pi$ and~$\pi'$ for all~$F \in \mullerfam_1$.
We show how to lift this qualitative construction to the setting of quantitative games by providing a quantitative reduction from quantitative Muller games to vertex-ranked~$\sup$-safety games.

\begin{lemma}
\label{lem:muller:reduction}
Let~$\game$ be a quantitative Muller game with~$n$ vertices.
There exists a memory structure~$\mem$ of size at most~$(n!)^3$ and a vertex-ranked~$\sup$-safety game~$\game'$ such that~$\game \reducesto^3_\mem \game'$.
\end{lemma}

\begin{proof}
Let~$\game = (\arena, \cost)$ with $\cost = \quantmuller(\mullerfam_0, \mullerfam_1)$.
We say that two play prefixes~$\pi$ and~$\pi'$ are~$\mullerfam_1$-equivalent if they end in the same vertex and if, for each~$F \in \mullerfam_1$, we have~$\acc_F(\pi) = \acc_F(\pi')$ and~$\score_F(\pi) = \score_F(\pi')$.
In this case, we write~$\pi \approx_{\mullerfam_1} \pi'$.
For each play prefix~$\pi$, we denote the~$\mullerfam_1$-equivalence-class of~$\pi$ by $\equivclass{\pi}{\approx_{\mullerfam_1}} = \set{\pi' \in V^* \mid \pi \approx_{\mullerfam_1} \pi'}$.
Furthermore, for each set~$\Pi \subseteq V^*$ of play prefixes we define the $\approx_{\mullerfam_1}$-quotient of~$\Pi$ as $(\quotient{\Pi}{\approx_{\mullerfam_1}}) = \set{\equivclass{\pi}{\approx_{\mullerfam_1}} \mid \pi \in \Pi}$.
For the sake of readability, we omit the index~$\mullerfam_1$ of the score-function and the index~$\approx_{\mullerfam_1}$ of the equivalence class for the remainder of this proof wherever possible without introducing ambiguity.

Let~$\plays_{\leq 2} = \set{\pi \in V^* \mid \score_{\mullerfam_1}(\pi) \leq 2}$ be the set of play prefixes whose score is at most two.
We define the set of memory states~$M = \left( \quotient{\plays_{\leq 2}}{\approx} \right) \cup \set{\bot}$,
the initial memory state~$m_\initmark = \equivclass{v_\initmark}{}$, and the update function~$\update$
as~$\update(\bot, v) = \bot$,
~$\update(\pi, v) = \equivclass{\pi v}{}$ 
if~$\score(\pi v) \leq 2$ and~$\update(\pi, v) = \bot$ otherwise.
We obtain~$\card{M} \in \bigo(\card{{\quotient{\plays_{\leq 2}}{\approx}}})$.
Since $\card{\quotient{\plays_{\leq 2}}{\approx}} \leq (n!)^3$ due to Neider, Rabinovich, and Zimmermann~\cite{NRZ14}, the memory structure~$\mem$ is of size at most~$(n!)^3$ as well.

A straightforward induction shows that this memory structure tracks the score of a play precisely as long as it does not exceed the value two on any prefix.
More formally, it satisfies the following invariant:
\begin{quote}
	Let~$\pi = v_0\cdots v_j$ be a play prefix in~$\game$ such that~$\score(v_0\cdots v_k) \leq 2$ for all~$k$ with~$0 \leq k \leq j$.
	Moreover, let~$\update^+(\pi) = \pi'$.
	Then~$\pi \approx \pi'$.
\end{quote}
Recall~$\safety(U) = \set{v_0v_1v_2\cdots \in V^\omega \mid \forall j \in \nats.\, v_j \notin U}$.
We define the vertex-ranked~$\sup$-safety game~$\game' = (\arena \times \mem, \suprankgame(\safety(V \times \set{\bot}), \rank))$, with~$\rank(v, \pi) = \score(\pi)$ for all~$\pi \in \plays_{\leq 2}$, and~$\rank(v, \bot) = 3$.

Let~$\cost' =  \suprankgame(\safety(V \times \set{\bot}), \rank)$.
Clearly, the first two items of the definition of~$\game \reducesto^3_\mem \game'$ hold true.
It remains to show~$\cost(\rho) = \cost'(\ext(\rho))$ for all~$\rho$ with~$\cost(\rho) < 3$ and~$\cost'(\ext(\rho)) \geq 3$ for all other~$\rho$.

First, let~$\rho = v_0v_1v_2\cdots$ be some play with~$\cost(\rho) \leq 2$ and let~$\ext(\rho) = (v_0,m_0)(v_1,m_1)(v_2,m_2)\cdots$.
Then~$\score(v_0 \cdots v_j) \leq 2$ for all~$j \in \nats$.
Thus, due to the invariant above and the definition of~$\rank$, we obtain~$\rank(v_j,m_j) = \score(v_0 \cdots v_j)$ for all~$j \in \nats$, which implies~$\cost'(\ext(\rho)) = \cost(\rho)$.

Towards a proof of the latter statement, let~$\rho = v_0v_1v_2\cdots$ be a play with $\cost(\rho) \geq 3$ and let~$j$ be the minimal position such that~$\cost(v_0\cdots v_j) = 3$.
Since $\cost(v_0\cdots v_j) = \score_F(v_0\cdots v_j)$ for some~$F \in \mullerfam_1$ and since the score is at most incremented by one during each step, we obtain~$\score_F(v_0 \cdots v_{j-1}) = 2$, ~$\acc_F(v_0 \cdots v_{j-1}) = F \setminus \set{v_j}$, and~$v_j \in F$.
Let~$\update^+(v_0 \cdots v_{j-1}) = \pi'$.
Due to the invariant we obtain~$\score_F(\pi') = 2$ and~$\acc_F(\pi') = F \setminus \set{v_j}$.
Thus,~$\score_F(\pi' v_j) = 3$, hence $\ext(v_0 \cdots v_j) = (v_0, m_\initmark)\cdots (v_j, \bot)$, which implies~$\ext(\rho) \notin \safety(V \times \set{\bot})$, which in turn yields~$\cost'(\ext(\rho)) = \infty > 3$.
	\qed
\end{proof}

Thus, in order to solve a quantitative Muller game with respect to some~$b$, it suffices to solve a vertex-ranked $\sup$-safety game~$\game'$ with respect to~$b$.
Recall that this is only constructive if Player~$0$ wins~$\game'$ with respect to~$b < 3$, i.e., only in this case are we able to construct a strategy with cost at most~$b$ for her in~$\game$.
Otherwise, Theorem~\ref{thm:capped-reduction} yields that there exists a strategy of cost~$\infty$ for Player~$1$ in~$\game$, but we cannot construct such a strategy from his strategy of cost greater than two in~$\game'$.
This is consistent with results of Neider et al.~\cite{NRZ14} and with the fact that Muller conditions are in a higher level of the Borel hierarchy than safety conditions.
Hence, qualitative Muller games cannot be reduced to safety games.

We can, however, solve the resulting vertex-ranked~$\sup$-safety game with respect to a given bound by solving a qualitative safety game as shown in Theorem~\ref{thm:vertex-ranked-games:direct:complexity}.
Using this reduction together with the framework of quality-preserving reductions, we obtain an upper bound on the complexity of solving quantitative Muller games with respect to some bound~$b$.

\begin{theorem}
	\label{thm:muller:exptime-membership}
	The following problem can be solved in time~$\bigo((n!)^3)$: \myquot{Given some quantitative Muller game~$\game$ with~$n$ vertices and some bound~$b \in \nats$, does Player~$0$ win~$\game$ with respect to~$b$?}
\end{theorem}

\begin{proof}
	Given~$\game$, we first construct the vertex-ranked~$\sup$-safety game~$\game'$ as shown in Lemma~\ref{lem:muller:reduction}.
	Recall that~$\game'$ contains at most~$(n!)^3$ vertices.
	Due to Theorem~\ref{thm:vertex-ranked-games:direct:complexity} and the fact that safety games can be solved in linear time in the number of vertices,~$\game'$ can indeed be solved in time at most~$\bigo((n!)^3)$ with respect to a given bound~$b$.
	\qed
\end{proof}

Analogously to the reasoning leading to Corollary~\ref{cor:vertex-ranked-games:sup:memory} on Page~\pageref{cor:vertex-ranked-games:sup:memory} and to Corollary~\ref{cor:vertex-ranked-games:lim:memory} on Page~\pageref{cor:vertex-ranked-games:lim:memory}, we are now also able to provide an upper bound on the size of strategies for Player~$0$ in quantitative Muller games.
Since both players have positional winning strategies in safety games, an application of Theorem~\ref{thm:reductions:strategy-lift} to Lemma~\ref{lem:muller:reduction} yields that if Player~$0$ has a strategy with cost at most three in a quantitative Muller game~$\game$ with~$n$ vertices, then she also has a strategy in~$\game$ with the same cost and of size at most exponential in~$n$.

Moreover, we can bound the complexity of the optimization problem for quantitative Muller games as follows:
Finding the minimal~$b$ such that Player~$0$ has a strategy of cost at most~$b$ in~$\game$ requires solving at most three safety games of size in~$\bigo((n!)^3)$.
Thus, the optimization problem for quantitative Muller games can be solved in factorial time.

%% file: applications-resilience.tex
\subsection{Fault Resilient Strategies for Safety Games}
\label{sec:application:resilience}

We now demonstrate the flexibility and versatility of vertex-ranked games in their own right.
To this end, we consider the problem of synthesizing a controller for a reactive system that is embedded into some environment.
This setting is typically modeled as an infinite game in which Player~$0$ and Player~$1$ take the roles of the controller and of the environment, respectively.
Here, we consider safety games, i.e., we assume that the specification for the controller is given as a game in which it is the aim of Player~$0$ to keep the play inside a safe subset of the vertices.

Dallal, Neider, and Tabuada~\cite{DallalNeiderTabuada16} argue that this setting is not sufficiently expressive to correctly model a real-world scenario, since it assumes that Player~$0$ can accurately predict the effect of her actions on the state of the system.
In a realistic setting, in contrast, faults may occur, i.e., an action chosen by a controller may be executed incorrectly, or it may not be executed at all.

In order to model such faults, Dallal, Neider, and Tabuada introduce arenas with faults~$\arena_F = (V, V_0, V_1, E, F, v_\initmark)$, which consist of an arena without faults~$(V, V_0, V_1, E, v_\initmark)$ and a set of faults~$F \subseteq V_0 \times V$.
In such an arena, whenever it is the turn of Player~$0$, say at vertex~$v$, a fault $(v, v') \in F$ may occur, resulting in the play continuing in vertex~$v'$ instead of that chosen by Player~$0$.
Moreover, Dallal et al.\ consider safety conditions, i.e., qualitative winning conditions of the form~$\safety(U)= \set{ v_0v_1v_2\cdots \mid \forall j \in \nats.\, v_j \notin U}$ for some~$U \subseteq V$.
Hence, it is the aim of Player~$0$ to keep the play outside the ``unsafe'' set of vertices~$U$.
If the play enters the set~$U$, it is declared winning for Player~$1$.
The task at hand is to compute a fault-resilient strategy for Player~$0$ that forces the play to remain inside~$V \setminus U$ and that can ``tolerate'' as many faults as possible.

Safety games without faults are solved by a simple attractor construction:
As soon as the play enters~$W_1 = \att{1}{}{U}$, Player~$1$ can play consistently with his attractor strategy towards~$U$ in order to win the play.
Thus, it is the aim of Player~$0$ to keep the play inside $W_0 = V \setminus W_1$.

Dallal, Neider, and Tabuada solve the problem of computing fault-resilient strategies for safety games by adapting the classic algorithm for solving safety games to this setting.
In doing so, they obtain a value~$\val(v)$ for each vertex~$v$ that denotes the minimal number of faults that need to occur in order for the play to reach~$W_1$, if Player~$0$ plays well.
Furthermore, they show that~$\val$ can be computed in polynomial time in~$\card{V}$.
Finally, due to the existence of positional winning strategies for both players in safety games, they obtain~$\val(v) \in [n] \cup \set{\infty}$ for all~$v \in V$.
Then, a fault-resilient strategy for Player~$0$ is one that maximizes the minimal value~$\val(v)$ witnessed during any play.
Dallal, Neider, and Tabuada construct such a strategy on the fly during the computation of~$\val(v)$~\cite{DallalNeiderTabuada16}.

This task can, however, easily be reframed as a vertex-ranked game in the arena $\arena = (V, V_0, V_1, E, v_\initmark)$, which we obtain from~$\arena_F$ by omitting the faults.
In that game, we assign to each vertex the rank~$\rank(v) = \card{V} - \val(v)$ if~$\val(v) \in [n]$ and~$\rank(v) = 0$ otherwise, i.e., if~$\val(v) = \infty$.
Then, Player~$0$ has a strategy with cost at most~$b$ in~$\game' = (\arena, \suprankgame(\safety(U), \rank))$ if and only if she has a winning strategy in the original safety game with faults that tolerates at least~$\card{V} - b$ faults.

This formulation as a vertex-ranked game enables further study of games in arenas with faults.
Here, we require the winning condition to be a safety condition in order to compute~$\val(v)$.
In recent work, we have shown how to compute this value for more complex qualitative winning conditions~\cite{NeiderWeinertZimmermann18}.
If~$\val(v)$ is effectively computable for a given qualitative winning condition, one can easily obtain fault-resilient strategies by formulating the task as a vertex-ranked game as demonstrated.

Finally, the formulation as a vertex-ranked game yields a method to compute eventually-fault-resilient strategies, i.e., strategies that are resilient to a large number of faults after a finite ``start-up'' phase.
In order to obtain such strategies, it suffices to view the resulting vertex-ranked games as a $\limsup$-game instead of a~$\sup$-game and to solve it optimally as described in Section~\ref{sec:vertex-ranked-games:limsup}.

%% file: conclusion.tex
\section{Conclusion}

In this work, we have lifted the concept of reductions, which has yielded a multitude of results in the area of qualitative games, to quantitative games.
We have shown that this novel concept exhibits the same useful properties for quantitative games as it does for qualitative ones and that it furthermore retains the quality of strategies.

Additionally, we have provided two very general types of quantitative games that serve as targets for quantitative reductions, namely vertex-ranked $\sup$ games and vertex-ranked~$\limsup$-games.
For both kinds of games we have shown a polynomial overhead on the complexity of solving them with respect to some bound, on the memory necessary to achieve a given cost, and on the complexity of determining the optimal cost that either player can ensure.

Finally, we have demonstrated the versatility of these tools by using them to solve quantitative request-response games and quantitative Muller games and by showing how to solve the problem of computing fault-resilient strategies in safety games via vertex-ranked games.
This last formulation enables a general study of games with faults, even in the presence of more complex winning conditions than the safety condition considered by Dallal et al.~\cite{DallalNeiderTabuada16} and in this work.
We are currently investigating how to leverage vertex-ranked games for the synthesis of fault-resistant strategies in parity games.

Further research continues in two additional directions:
Firstly, while the framework of quantitative reductions and vertex-ranked games yields upper bounds on the complexity of solving quantitative games, it does not directly yield lower bounds on the complexity of the problems under investigation.
Consider, for example, the threshold problem for parity games with costs, which is~$\pspace$-complete~\cite{WeinertZimmermann17}.
It is possible to reduce this problem to that of solving a vertex-ranked parity game of exponential size and linearly many colors similarly to the reduction presented in this work, which yields an~$\exptime$-algorithm.
It remains open how to use quantitative reductions to obtain an algorithm for this problem that only requires polynomial space.

Secondly, another goal for future work is the establishment of an analogue to the Borel hierarchy for quantitative winning conditions.
In the qualitative case, this hierarchy establishes clear boundaries for reductions between infinite games, i.e., a game whose winning condition is in one level of the Borel hierarchy cannot be reduced to one with a winning condition in a lower level.
Also, each game with a winning condition in the hierarchy is known to be determined~\cite{Martin75}.
To the best of our knowledge, it is open how to define such a hierarchy for quantitative winning conditions which exhibit similar properties.

\paragraph*{Acknowledgements}
I would like to thank Martin Zimmermann for many fruitful discussions.
Furthermore, I would like to thank the anonymous reviewers for their thorough reviews and constructive comments.